\documentclass[twocolumn,epsfig,graphics,noshowpacs,floatfix,nofootinbib]{revtex4}

\usepackage{amsmath,amsfonts,amssymb,graphics,graphicx,epsfig,color,times,bbm}
\usepackage{amsthm}
\usepackage{psfrag}
\usepackage{braket}
\AtBeginDocument{\usepackage{booktabs}}
\graphicspath{{figures/}}
\usepackage[hidelinks]{hyperref}
\hypersetup{colorlinks=false}
\usepackage{array}
\newcolumntype{P}[1]{>{\centering\arraybackslash}p{#1}}
\newcolumntype{M}[1]{>{\centering\arraybackslash}m{#1}}

\usepackage{graphicx}
\usepackage{mathtools}
\usepackage{soul}
\usepackage[normalem]{ulem}

\newtheorem{definition}{Definition}
\newtheorem{theorem}{Theorem}

\newcommand{\zz}{\mathbbm{Z}}

\newcommand{\gauge}{\mathcal{G}}
\newcommand{\stabilizer}{\mathcal{S}}

\newcommand{\fusenet}{{G}}
\newcommand{\fusenetV}{{G}_V}
\newcommand{\fusenetE}{{G}_E}
\newcommand{\region}{A}
\newcommand{\resourcegroup}{\mathcal{R}}
\newcommand{\fusiongroup}{{\mathcal{F}}}
\newcommand{\checkgroup}{{\mathcal{C}}}
\newcommand{\complex}{\mathcal{L}}
\newcommand{\complexF}{\mathcal{L}}
\newcommand{\complexR}{\mathcal{L}_R}

\usepackage[framemethod=tikz]{mdframed}
\usepackage{hyperref}

\definecolor{warning_bgcol}{RGB}{252,248,229}
\definecolor{warning_textcol}{RGB}{111,89,54}
\definecolor{warning_linecol}{RGB}{252,248,229}
\definecolor{danger_bgcol}{RGB}{239,223,222}
\definecolor{danger_textcol}{RGB}{128,60,57}
\definecolor{danger_linecol}{RGB}{239,223,222}
\definecolor{success_bgcol}{RGB}{224,237,216}
\definecolor{success_textcol}{RGB}{68,104,60}
\definecolor{success_linecol}{RGB}{224,237,216}
\definecolor{info_bgcol}{RGB}{220,237,246}
\definecolor{info_textcol}{RGB}{58,100,126}
\definecolor{info_linecol}{RGB}{220,237,246}

\mdfdefinestyle{box_style}{
  skipabove=.7\baselineskip,
  skipbelow=.7\baselineskip,
  innertopmargin=.65\baselineskip,
  innerbottommargin=.65\baselineskip,
  innerleftmargin=.5\baselineskip,
  innerrightmargin=.5\baselineskip,
  splittopskip=1.5\baselineskip,
  splitbottomskip=\baselineskip,
  roundcorner=.3\baselineskip
}
\mdfdefinestyle{warning_style}{
  style=box_style,
  backgroundcolor=warning_bgcol,
  linecolor=warning_linecol,
  fontcolor=warning_textcol,
}
\mdfdefinestyle{success_style}{
  style=box_style,
  backgroundcolor=success_bgcol,
  linecolor=success_linecol,
  fontcolor=success_textcol,
}
\mdfdefinestyle{danger_style}{
  style=box_style,
  backgroundcolor=danger_bgcol,
  linecolor=danger_linecol,
  fontcolor=danger_textcol,
}
\mdfdefinestyle{info_style}{
  style=box_style,
  backgroundcolor=info_bgcol,
  linecolor=info_linecol,
  fontcolor=info_textcol,
}

\begin{document}

\title{Fault-tolerant complexes}

\newcommand*\leadauthor{\thanks{Lead authors. All authors are listed alphabetically. Correspondence: naomi@psiquantum.com}}

\author{H\'ector Bomb\'in}
\author{Chris Dawson}
\author{Terry Farrelly}
\author{Yehua Liu} \leadauthor 
\author{\\ Naomi Nickerson} \leadauthor 
\author{Mihir Pant}
\author{Fernando Pastawski}
\author{Sam Roberts}

\affiliation{PsiQuantum, Palo Alto}
\date\today

\begin{abstract}
Fault-tolerant complexes describe surface-code fault-tolerant protocols from a single geometric object. We first introduce \emph{fusion complexes} that define a general family of fusion-based quantum computing (FBQC) fault-tolerant quantum protocols based on surface codes. We show that any 3-dimensional cell complex where each edge has four incident faces gives a valid fusion complex. This construction enables an automated search for fault tolerance schemes, allowing us to identify 627 examples within a moderate search time. We implement this using the open-source software tool \texttt{Gavrog} and present threshold results for a variety of schemes, finding fusion networks with higher erasure and Pauli thresholds than those existing in the literature. We then define more general structures we call \emph{fault-tolerant complexes} that provide a homological description of fault tolerance from a large family of low-level error models, which include circuit-based computation, floquet-based computation, and FBQC with multi-qubit measurements. This extends the applicability of homological descriptions of fault tolerance, and enables the generation of many new schemes which have not been previously identified. We also define families of fault-tolerant complexes for color codes and 3d single-shot subsystem codes, which enables similar constructive methods, and we present several new examples of each.  
\end{abstract}

\maketitle

\section{Introduction}

Useful quantum computing demands quantum error correction to enable quantum algorithms to be executed without being overcome by noise. Fusion-based quantum computing (FBQC) provides a natural model for fault tolerance in photonic architectures~\cite{bartolucci2023fusion}. In FBQC, fault tolerance is achieved by constructing fusion networks made up of constant-sized resource states and performing entangling projective fusion measurements~\cite{browne2005resource} on qubits from different resource states.
However, the additional structure of checks and logical membranes needed to guarantee topological fault-tolerance is not explicit in the fusion network, which makes the task of defining new schemes laborious. To date there only existed a handful of examples of schemes constructed in an ad hoc way~\cite{duan2005efficient,gimeno2015three,fukui2018high,pant2019percolation,bartolucci2021creation,sahay2022tailoring,paesani2022high,lee2023parity,pankovich2023high}. In other computational models, surface codes are known to be defined for a family of geometries, for example, 2d surface codes exist for every 2-dimensional cell complex~\cite{kitaev2006anyons}, and in measurement-based fault tolerance, a fault-tolerant cluster state can be defined for any 3-dimensional cell complex~\cite{raussendorf2007topological, nickerson2018measurement,newman2020generating}. These generalizations have allowed for the exploration of many surface geometries~\cite{fujii2012error,hyperbolic_codes,chamberland2020topological}. How, then, should one look to define the family of possible surface-code fusion networks? What resource states and measurements can be combined to produce the necessary structure of correlations? 

We find that a large family of fault-tolerant fusion networks can be defined through a simple construction, which we call \emph{fusion complexes}. A fusion complex is 3-dimensional cell complex in which the vertices correspond to resource states and the edges correspond to fusions. We will show that if the cell complex has 4 faces incident at every edge in the bulk then it defines a surface-code fault tolerance scheme. This construction simplifies the definition of FBQC fault tolerance schemes, capturing the measurements, the structure of resource states, as well as the check structure in a single geometric picture. Here by \emph{checks} (which are defined more carefully later), we refer to products of measurement outcomes that give deterministic values in the absence of errors (and are also known as `detectors' in some parts of the literature). Fusion complexes are introduced in Section~\ref{sec:fusion_complexes}.
The most immediately practical application of the fusion complex construction is that it enables an automated search for new fusion network protocols, which we cover in Section~\ref{sec:automated_search}. Building on the techniques introduced in~\cite{newman2020generating}, we automate the search for fusion complexes and identify 627 examples within a moderate search time. The size and complexity of the search space is parametrically controlled, and it can be broadened to find further examples as desired. 

We next show how the fusion complex construction can be generalized to apply to other computational paradigms in what we call~\emph{fault-tolerant complexes}, where the resource states of FBQC are replaced by other operational primitives, such as those relevant for circuit based quantum computing (CBQC)~\cite{raussendorf2007topological,wang2011surface,Dennis_2002,fowler2012surface,corcoles2015demonstration,google2023suppressing,chamberland2022building,bourassa2021blueprint,mcewen2023relaxing,geher2023tangling}, measurement-based quantum computing (MBQC)~\cite{raussendorf2005long, raussendorf2007topological, raussendorf2006fault,foliated_codes,nickerson2018measurement,brown2020universal,newman2020generating}, floquet-based quantum computing (FloBQC)~\cite{hastings2021dynamically, paetznick2023performance, haah2022boundaries, gidney2022pair, kesselring2022anyon, davydova2023floquet, ZXhappyfamilies, paesani2022high, aasen2022adiabatic, townsend2023floquetifying, dua2023engineering}, or generalized FBQC where multi-qubit fusion measurements are used. We discuss this in Section~\ref{sec:homology}. 

It has previously been shown that under some error models it is possible to represent all the fundamental objects---qubits, checks, errors, logical operators---through a single geometrical object: a 3-dimensional cell complex~\cite{raussendorf2005long, raussendorf2007topological, raussendorf2006fault}, for example in MBQC with single qubit measurement errors, where the homological approach for fault tolerance was first proposed~\cite{raussendorf2005long}.
These relationships provide a very useful tool for analysis, simulation, as well as generation of new models.
However, when errors are accounted for in primitive operations, such as FBQC or circuit-level error models, existing methods have not been able to relate all the check, logical and error structure to a single geometrical representation. Instead error models are treated using correlations~\cite{raussendorf2006fault,chamberland2020topological}.  

By contrast, the fault-tolerant complex construction allows us to extend homological representations to a much wider class of error models, including (in many cases) circuit-level error models. This extended description provides a single geometric construction containing all the check structure information.
We find there is a fundamental relationship between the size and structure of the primitives and the structure of the check operators. The fault-tolerant complex construction can also enable generative search for protocols from low-level operations.

Fault-tolerant cell complexes have a wider applicability beyond surface codes. We show in the appendices how similar definitions can be found for color codes and subsystem codes (specifically 3d subsystem toric codes), again enabling generative search for new codes and fault-tolerant protocols.




\begin{figure}
    \centering
    \includegraphics[width=\columnwidth]{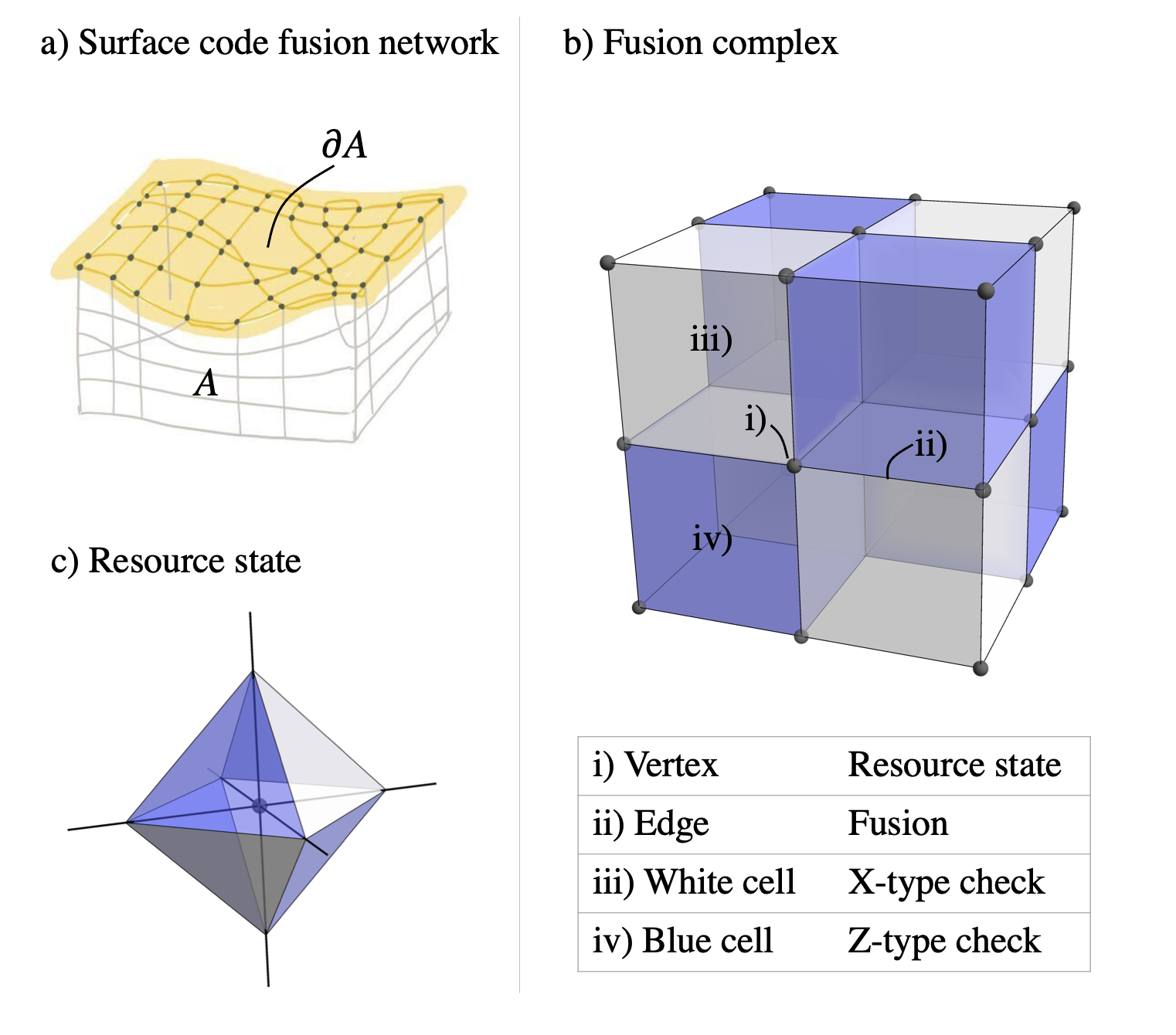}
    \caption{a) A region, $A$, of a fusion network is shown with a boundary $\partial A$. When the fault tolerance protocol inside $A$ has been performed a surface code state is left on the remaining boundary qubits. b) An example, the cubic fusion complex. c) A representation of the resource state associated with a vertex as an inflated volume. The state is represented as a surface code (in the plaquette representation) on the surface of an octahedron.  }
    \label{fig:fusion_complex}
\end{figure}

\section{Fusion Complexes}
\label{sec:fusion_complexes}

\subsection{Surface code fusion networks}
In fusion-based quantum computing~\cite{bartolucci2023fusion} fault tolerance is achieved by building a \emph{fusion network} which defines a set of measurements (\emph{fusions}) made on a collection of constant-sized entangled states (\emph{resource states}). The resource states are stabilizer states, described by a stabilizer group, $\resourcegroup$, and the fusion measurements by the fusion group, $\fusiongroup$, as described in Ref~\cite{bartolucci2023fusion}. When all fusions are 2-way measurements a fusion network can be conveniently represented as a graph where the vertices represent resource states and the edges represent fusion measurements between them. Two examples of fault-tolerant fusion networks were shown in~\cite{bartolucci2023fusion}. An analysis of these schemes shows that they correspond to fault-tolerant fusion networks encoding surface codes (which is made more precise through the following two definitions). But what other schemes are possible for fault-tolerant fusion networks? We are motivated to seek a family of such schemes. 

First let us introduce the notion of a \emph{surface code fusion network} to describe the general properties we are seeking. To do so, we must first define what we mean by a surface code. Note that from here on, the term \textit{surface code} refers to 2d surface code.

\begin{definition}\label{def:surface_code}(Surface code)
  A stabilizer code $\mathcal{S}$ is a surface code, if there exists a 2d cell complex $\complex = \{F, E, V\}$ where vertices are 4-valent and faces are 2-colorable, such that $\mathcal{S} = \langle X_{f_a}, Z_{f_b} ~|~ \forall f_a \in F_a, f_b \in F_b \rangle$, where $X_{f_a}$ is a product of Pauli-$X$ operators on all vertices belonging to the face $f_a$, $Z_{f_b}$ is a product of Pauli-$Z$ operators on all vertices belonging to the face $f_b$, and $F_a$ and $F_b$ denote the set of faces of each color. 
\end{definition}

We refer to this presentation of the surface code as the \textit{plaquette} version of the surface code~\cite{wen2003quantum}. We note that the 4-valence condition implies a local bi-colorability of the faces, and for certain complexes, also a global bi-colorability~\cite{anderson2013homological}\footnote{Although Def.~\ref{def:surface_code} requires bi-colorability of the faces, some twisted global boundary conditions may remove global bi-colorability whilst still leaving a locally valid surface code.}. See Def.~\ref{def:kitaev_surface_code} in Appendix~\ref{app:surface_codes} for an equivalent, but perhaps more familiar definition of a surface code due to Kitaev~\cite{kitaev2003fault}. Errors in a surface code fusion network can be represented by a \emph{syndrome graph}\footnote{\emph{Syndrome graphs} represent the relationship between errors and the checks they flip. Checks are represented by vertices, and errors are represented by edges, connecting two vertices whenever they flip the corresponding checks.}.
We consider any stabilizer code that is equivalent to a surface code under local Clifford operations to also be a surface code (such as the Wen plaquette model~\cite{wen2003quantum}).

To define a surface code-fusion network, we restrict ourselves to the case where all fusions are 2-qubit Bell-basis measurements\footnote{Note that FBQC can also make use of multi-qubit projective measurements, and the more general picture introduced in Section~\ref{sec:homology} can be used to construct these protocols.} $f = \{XX, ZZ\}$. This allows the fusion network to be represented by a graph $\fusenet = (\fusenetV, \fusenetE)$, whose vertices $\fusenetV$ correspond to resource states and edges $\fusenetE$ to fusions. Note in particular that each edge supports two qubits, one from each neighbouring resource state. 

\begin{definition}(Surface-code fusion network)
A fusion network is a \textit{surface-code fusion network} if after performing fusions on any region $\region\subset \fusenetV$, the unmeasured qubits on the boundary of the region  $\partial \region$ are described by a surface code (up to signs of the stabilizers). 
\end{definition}

Here, performing fusions on a region $\region\subset \fusenetV$ means performing all fusions between resource states contained within that region. The boundary $\partial \region$ refers to all un-fused qubits of resource states within $\region$. In particular, this implies all resource states must be surface codes themselves (corresponding to the region of a single vertex). 

An important property of surface-code fusion networks, is that un-fused qubits in the network at any moment in time (a time-slice) during a quantum computation will always be in a surface-code state (up to signs of the stabilizers). An illustration of this is shown in Fig.~\ref{fig:fusion_complex}a).
The surface code, however, need not be the same one from one layer to the next. The slices could also pass through the network in a spatial dimension or around a closed surface, but regardless of the orientation, if we inspect the state of the qubits along that slice they should be in a surface code state.

\subsection{Definition of a fusion complex}

We can now introduce fusion complexes which allow us to define a large family of surface code fusion networks along with a complete description of their check structures. We continue to consider the case where all fusions are 2-qubit Bell-basis measurements. A fusion complex, is a 3-dimensional cell complex with the constraint that every edge has exactly four incident faces.

\begin{definition}\label{def:fusion_complex} (Surface-code fusion complex)
  A 3-dimensional cell complex, $\complexF=\{C,F,E,V\}$, is a surface-code fusion complex if every edge has exactly four incident faces. 
\end{definition}

\begin{theorem}\label{thm:fusion-complex}
Every surface-code fusion complex defines a surface-code fusion network.
\end{theorem}

Before proving this statement, we first describe how the elements of the (surface-code) fusion complex correspond to the objects of a fusion network. Each vertex, $v\in V$ corresponds to a resource state containing one qubit for each edge incident to $v$. Each edge, $e \in E$, corresponds to a 2-qubit fusion measurement between the two resource states associated with the vertices at the endpoints of the edge. Consequently the 1-skeleton of the fusion complex is a graph $\fusenet$ describing the fusion network. An example of a cubical fusion complex is shown in Fig.~\ref{fig:fusion_complex}b). This example is the 6-ring scheme introduced in~\cite{bartolucci2023fusion}, up to local modification of resource-state qubits by Hadamard operations. The fusion network can be thought of as a set of resource states and fusions in 3d space-time, and the fusions can be performed in any order.

The fusion complex contains additional information in that each cell, $c\in \checkgroup$, corresponds to a check operator. Here, in the FBQC setting, a check operator is a Pauli operator that is a member of both the resource group, $\resourcegroup$, and the group generated by all the fusion measurements, $\fusiongroup$. That is, check operators are members of $\resourcegroup \cap \fusiongroup$. Note that here and throughout, we use the term \textit{cell} to refer to a 3-dimensional cell.

Here we make an important note on conventions for naming and how fault tolerance schemes are represented. 
Firstly, we have the convention that all fusion measurements are Bell state measurements of $XX$ and $ZZ$, and consequently that all check operators made up entirely of either $XX$ or $ZZ$ measurements. Other variations can be found by applying Hadamards (or other local Cliffords) in the protocol which can produce checks of mixed type (as is the case in~\cite{bartolucci2023fusion, bombin2023logical}). Secondly, we use the terminology `X-type' and `Z-type' to refer to the different flavors of checks, or simply `X checks' and `Z checks'. Note that we are intentionally avoiding the use of the commonly used language of `primal' and `dual'. This is because the X-type and Z-type check structures in fault-tolerant protocols are not, in general, dual to one another and one of our goals here is to clarify their relationship. When `primal' and `dual' are used it will be instead to refer to the relationship between two cell complexes. 

Returning now to the fusion complex; by construction (Def.~\ref{def:fusion_complex}) the cells are always bi-colorable\footnote{Strictly speaking, the definition only implies that 3-cells are locally bi-colorable since four 3-cells always meet at an edge. We assume throughout that the complex also has a global bi-coloring, even though this is not strictly necessary, as any obstructions from extending the local bi-coloring to a global bi-coloring (i.e., any region that breaks bi-colorability) can be treated as a domain wall. See Appendix~\ref{app:features}.}. The two colors of cells correspond to the X-type and Z-type check operators. For each X cell (Z cell) a check exists which takes the product of the $XX$ ($ZZ$) outcomes from the fusion on each edge contained in the cell. As each edge has exactly four incident faces, there are also exactly four cells meeting at each edge. Thus each fusion outcome (either $XX$ or $ZZ$) is supported in two check operators.

The resource state at any given vertex, $v$, is a stabilizer state defined by the structure of the surrounding cells and edges. The resource state can always be represented as a surface code on the surface of a topological sphere. An example for the cubical cell complex is shown in Fig.~\ref{fig:fusion_complex}c). In this example the resource state is defined on the surface of an octahedron. A simple way to visualize the resource state is by `inflating' the vertex in the fusion complex to create a volume. This volume has qubits on its vertices and bi-colorable faces that define the resource state stabilizers. Where the inflated volume intersects an X-type (Z-type) cell the resource state has an X-type (Z-type) plaquette for its face. Since the code is on the surface of a sphere this representation of the stabilizers gives an overcomplete-generating set. If there are $n$ qubits in the resource state then there are $n+2$ plaquette stabilizers on the sphere (with one X-type and one Z-type plaquette being a product of the others). 

A check operator exists for each cell of the fusion complex. In FBQC the check group is defined as $\checkgroup =\resourcegroup \cap \fusiongroup$, where $\resourcegroup$ is the stabilizer group of all resource states, and $\fusiongroup$ is the fusion group. A check exists wherever there is an element of $\resourcegroup$ that exists in $\fusiongroup$. We can verify that for each cell, $c$, of the complex there exists an element of the resource group, $r_c \in \resourcegroup$, that is identical to an element of the fusion group, $f_c \in \fusiongroup$. If, for example, the cell is X-type then $r_c$ is the product of the X-type stabilizer generator from each vertex of the cell, and $f_c$ is the product of the XX fusion measurement operator for each edge of the cell (and similarly for the Z-type checks).

Since there are always four cells meeting at an edge each fusion is part of exactly four checks, and since the cells are bi-colorable there must always be two X-checks and two Z-checks. An error flipping a fusion measurement outcome will therefore result in exactly two flipped checks, the property of a surface code under a single qubit Pauli-$X$ or $Z$ error. 

\begin{proof}[Proof of Theorem \ref{thm:fusion-complex}]
To prove Theorem~\ref{thm:fusion-complex} one can use the structure of the checks. For any region $\region\subset \fusenetV$, the state on the qubits $\partial \region$ after performing fusions within $\region$ can be obtained (up to sign) by restricting the check operators to $\partial \region$. This is a surface code state, as every qubit belongs to two X-check and two Z-check generators and therefore meets the conditions of Def.~\ref{def:surface_code}, following from the fact that every edge of the fusion complex belongs to two X-check and two Z-check generators. 
\end{proof}

We will go onto explore the surface code properties of this construction more carefully in Section~\ref{sec:homology}.

\subsection{A family of fault tolerance protocols}
In summary, fusion complexes are defined by the simple constraint that each edge of the complex has four incident faces, and this gives a large family of surface-code fusion networks. 
It is worth noting that sometimes the interpretation as a fusion complex may not be apparent at first glance---for example, GHZ states (which can be used as resources for FBQC) can be interpreted as surface codes on a sphere. In Section~\ref{sec:automated_search}, we present an automated search method for fusion complexes. There are also constructive methods for generating fusion complexes which we give in Appendix~\ref{app:construction}. Beyond FBQC, we present fault-tolerant complexes in Section~\ref{sec:homology} which describe protocols for more general computational models. Furthermore, fusion complexes can be used to define 3d subsystem toric codes, provided they also have bi-colorable vertices, which we show in Appendix~\ref{app:subsystem_code}. Color code versions of the fusion complex also exist and are explored in Appendix~\ref{app:color_code}.




\section{Automated search with Gavrog and examples}
\label{sec:automated_search}

The most immediately practical application of the definition of fusion complexes (Def.~\ref{def:fusion_complex}) is to automate the generation of examples. This enables high-throughput automated design and simulation of FBQC protocols.

\subsection{Automated search method}
In Ref.~\cite{newman2020generating} Newman \textit{et al.}\  showed that computational tiling theory could be used to generate 3-dimensional cell complexes that correspond to fault-tolerant cluster states for MBQC. Here we show that a similar method can be used to generate fusion complexes. A crucial difference is that in the MBQC setting the cell complex defines a very large `pre-entangled' state. In FBQC we are also accounting for the additional step of the construction of entanglement from small resource states. To achieve this, cell complexes must be filtered to find those that satisfy the constraints in Def.~\ref{def:fusion_complex}.

To systematically generate fusion complexes, we first generate candidate 3d cell complexes with the open-source software \texttt{Gavrog}~\cite{gavrog} before filtering them with Def.~\ref{def:fusion_complex}. In practice, it is easier to check the definition in the \textit{dual picture}, where cells are replaced with vertices, faces with edges, and so on, as described in Section~\ref{sec:homology}. The surface-code fusion complex definition requires that all faces in the dual cell complex have four edges (or equivalently, four vertices). This is checked by first looping over all cells in a unit cell of the dual complex, then looping over all faces of the cell. We only keep the candidate cell complex if all dual faces have four vertices.

In \texttt{Gavrog}, cell complexes are denoted by their D-symbols, which are generalized from Delaney-Dress symbols, also called the extended Schl\"afli symbols in Ref.~\cite{newman2020generating}. A D-symbol fully specifies a periodic tiling by its reflection symmetries, see Ref.~\cite{gavrog, Delgado-Friedrichs2001, Delgado-Friedrichs2003} for a detailed description. For example, the D-symbol for the cubic cell complex is \texttt{<1 3:1,1,1,1:4,3,4>}.

A subtlety in the implementation of \texttt{Gavrog} is that it excludes 3d cell complexes with cells that only contain two faces. We refer to these as `ravioli-cells' because of their resemblance to the pasta variety where two identical faces are joined at their boundaries to create a volume. In fusion complexes there is no constraint that prevents these cells, and indeed they can result in good fault tolerance properties since they often represent low weight checks. We allow such degenerate cells in our systematic search for fusion complex candidates.

Cell-complex generation in \texttt{Gavrog} is bounded by the maximum size of the D-symbol. The size describes the number of chamber classes in the Barycentric subdivision of the 3d tiling and is denoted by a positive integer \cite{Delgado-Friedrichs2003}. For example, the size of D-symbol is 1 for the cubic complex. We found 53 fusion complexes with sizes smaller or equal to 8, which are listed in Table~\ref{tab:examples_longlist}, and in total 627 fusion complexes with sizes smaller or equal to 15. More examples can be found by increasing the max D-symbol size in the search. 

\subsection{Notable examples}

Table~\ref{tab:fusion_complexes} contains some notable examples of fusion complexes with their D-symbols, syndrome-graph check degrees (or number of edges in cells), and resource-state sizes (or number of faces in dual cells). These examples are depicted in Fig.~\ref{fig:fusion_complex_examples}, and their resource states are shown in Fig.~\ref{fig:resource_states}.

\begin{figure}
    \centering
    \includegraphics[width=\columnwidth]{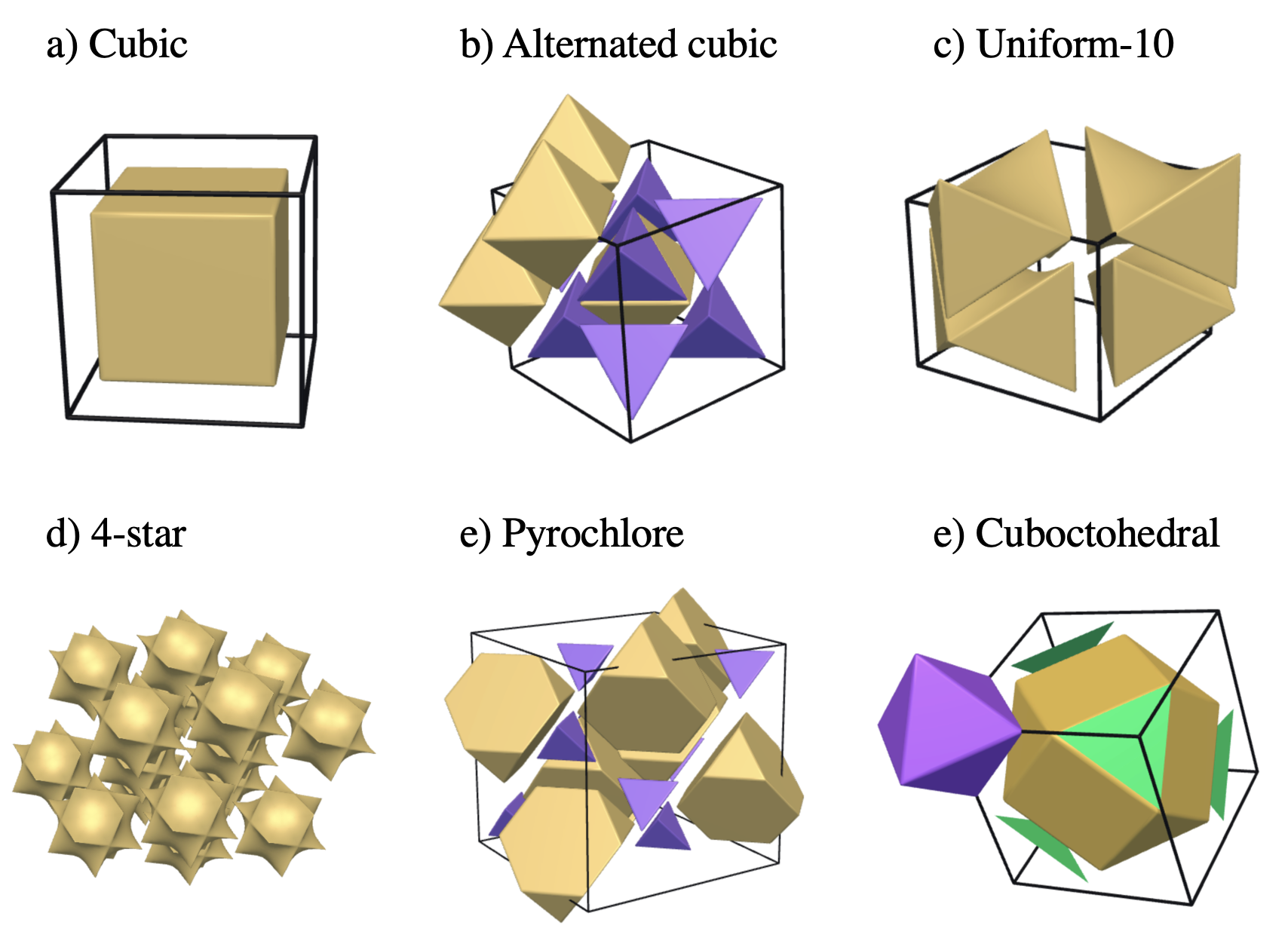}
    \caption{Graphical depictions of the fusion complexes listed in Table~\ref{tab:fusion_complexes}. Vertices correspond to resource states, edges to fusions, and cells to check operators. Note that the cells are colored for visual aid, and the color is not reflective of the cell being $X$ or $Z$ type. Images produced using \texttt{webGavrog}~\cite{webgavrog}.}  
    \label{fig:fusion_complex_examples}
\end{figure}

\begin{figure}
    \centering
    \includegraphics[width=\columnwidth]{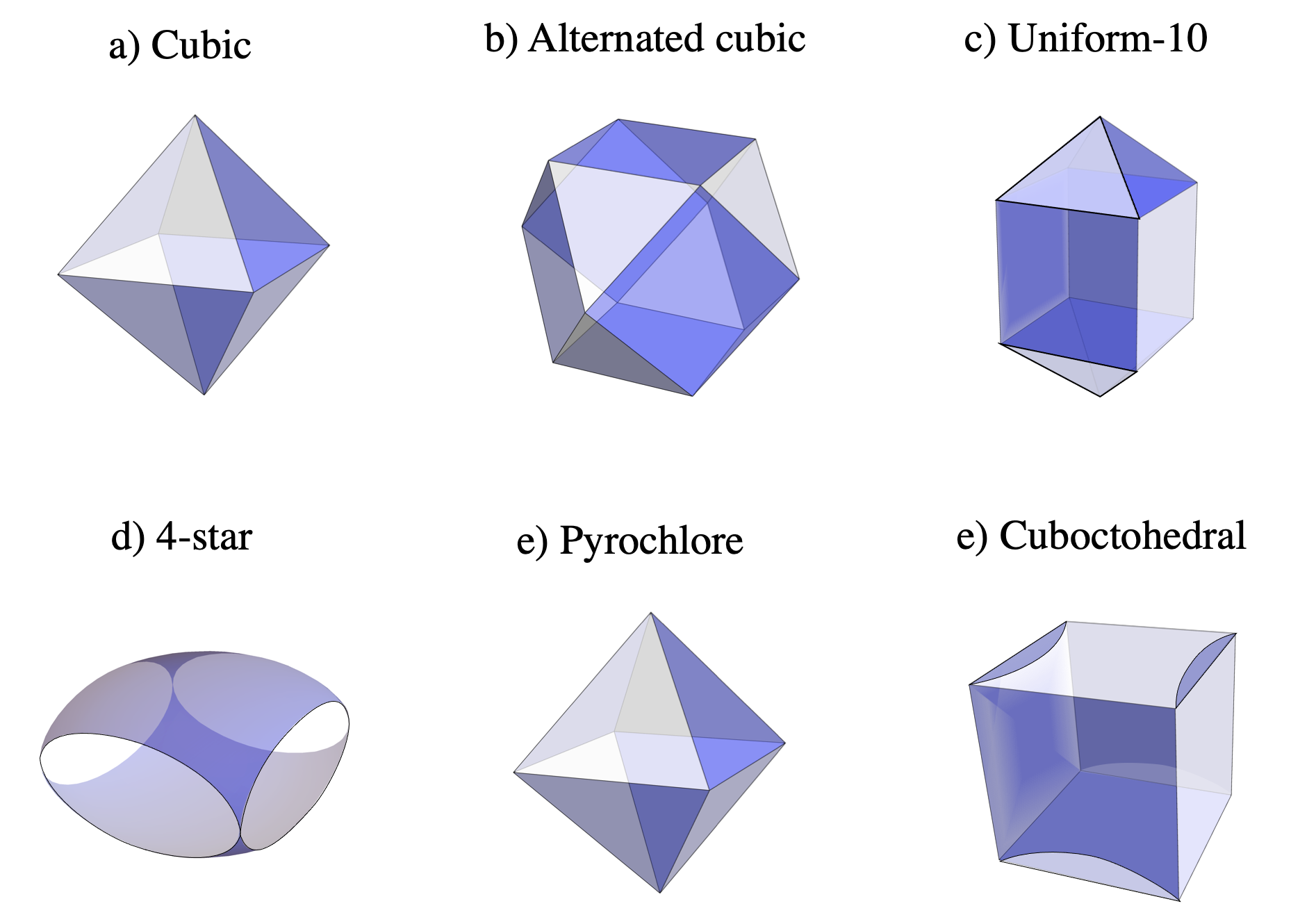}
    \caption{Graphical depictions of the resource states for each fusion complex listed in Table~\ref{tab:fusion_complexes}. Each resource state is depicted as a surface code on a closed surface in the plaquette representation. Each vertex represents a qubit and each face of the polyhedra represents a stabilizer of the resource state. There are two face colors to represent the X-type and Z-type stabilizers.  }
    \label{fig:resource_states}
\end{figure}

\begin{figure}
    \centering
    \includegraphics[width=\columnwidth]{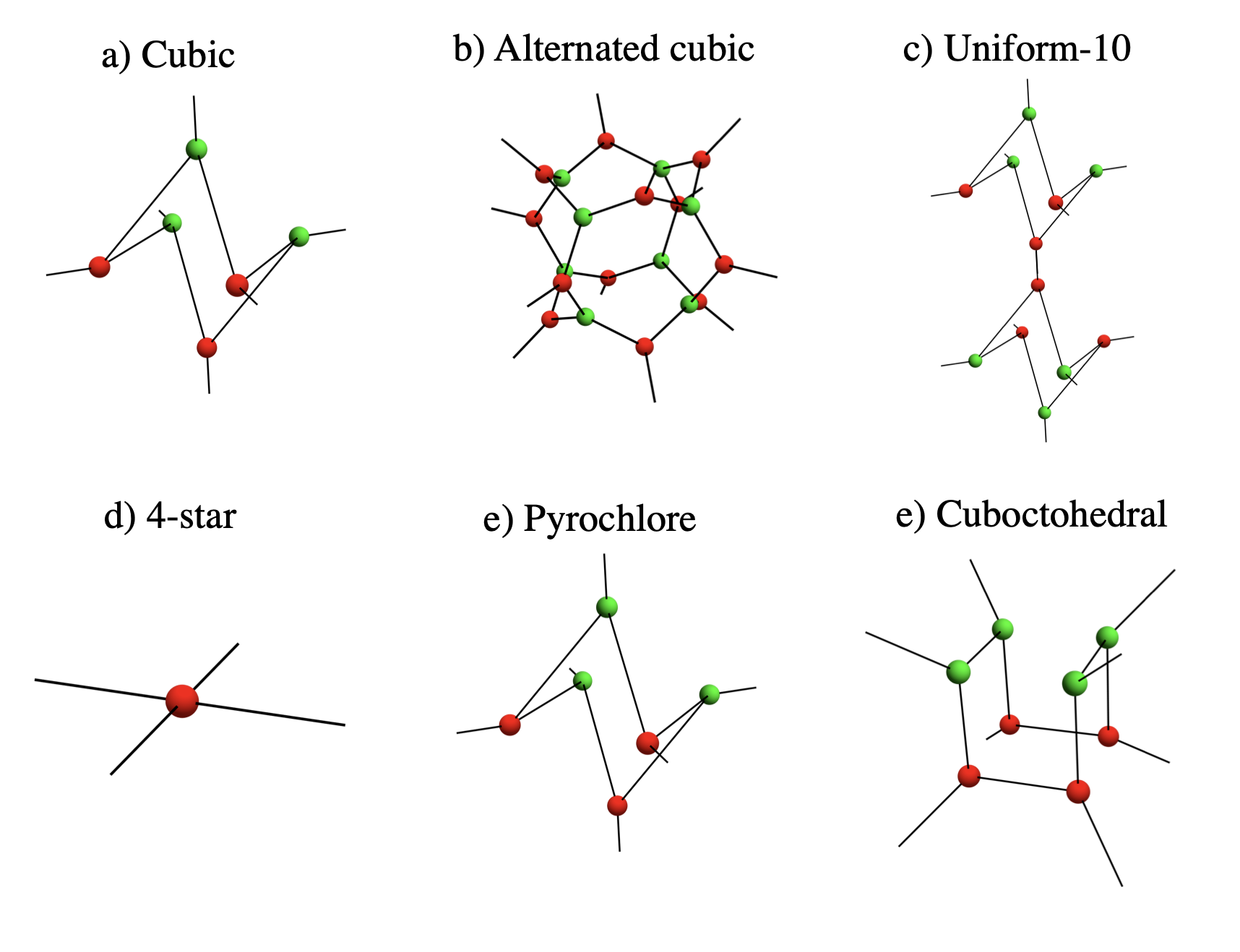}
    \caption{ZX-diagram representations of the resource states for each fusion complex listed in Table~\ref{tab:fusion_complexes}. We use the standard ZX-diagram notation here, for which Refs.~\cite{coecke2011interacting,coecke2011interacting,coecke2018picturing} provide thorough introductions, and Ref.~\cite{ZXhappyfamilies} covers in the setting of fault tolerance. Quantum operations are represented as a network of green ($Z$-type) and red ($X$-type) spiders. Spiders can be thought of as describing stabilizer-state projections, with each n-port spider describing n stabilizer generators on n qubits. The stabilizer generators described by $Z$ ($X$) spiders are $X^{\otimes n}$ ($Z^{\otimes n}$) and all pairwise $Z^{\otimes 2}$ ($X^{\otimes 2}$) operators. Multiple spiders can be connected to form composite ZX diagrams. Each connection between two spiders corresponds to a bell state projection, i.e., the two qubits $i$ and $j$ corresponding to the two ports are identified via the projection $Z_i Z_j = X_i X_j =+1$. Here our ZX diagrams represent states, with the outer legs representing the qubits of the state. ZX diagrams are not unique, there are many possible ZX diagrams that correspond to each resource state, here we show one illustrative example for each case. Note that the 4-star network is made up of both green and red spiders, but only one here is shown for simplicity. }
    \label{fig:ZX_resources}
\end{figure}

\textbf{Cubic.} The cubic fusion complex has constant syndrome-graph check degree 12 and resource-state size 6, and corresponds to the 6-ring fusion network introduced in~\cite{bartolucci2023fusion}. This is the example most closely related to circuit-based implementations of the standard square surface-code as described in Ref.~\cite{ZXhappyfamilies}. 

\textbf{Alternated cubic.} This complex uses a 12 qubit resource state and unlike the more symmetric cubic complex, it has checks of two different sizes. It is also asymmetric between $X$ and $Z$ checks (and syndrome graphs). 

\textbf{4-star.} This complex has resource states which are all 4-GHZ states, it corresponds to the fusion network introduced in Ref.~\cite{bartolucci2023fusion}. 
    
\textbf{Pyrochlore}. The Pyrochlore fusion complex has the same resource state size and the same average check degree as the 6-ring fusion complex. However there is lower symmetry in the structure of the checks.

\textbf{Cuboctahedral}. All resource states contain 8 qubits. The XX and ZZ syndrome graphs are isomorphic (under periodic boundary conditions), with 2/3 of the checks having degree 3, 1/6 having degree 12, and 1/6 having degree 24. The average check degree is 8. Allowing for degenerate cells is crucial for discovering this fusion complex in the systematic \texttt{Gavrog} search, due to the appearance of `ravioli-cells'.

\begin{table*}[t]
\begin{tabular}{lllllll}
\hline
Name & D-symbol & C & R & $p^\star_\mathrm{erasure}$ & $p^\star_\mathrm{flip}$ \\ \hline
Cubic (Ref.~\cite{bartolucci2023fusion}) & \scalebox{1}[1.0]{\tt\scriptsize<1 3:1,1,1,1:4,3,4>} & $12$ & $6$ & 12\% & 1\% \\ \hline
Alternated-cubic & \scalebox{1}[1.0]{\tt\scriptsize<2 3:1 2,1 2,1 2,2:3 3,3 4,4>} & $2\times6+12$ & $12$ & 12\%, 25\% & 1\%, 2.9\% \\ \hline
Uniform-10 & \scalebox{0.8}[1]{\tt\scriptsize<5 3:1 3 5,2 3 4 5,1 4 5,1 2 3 5:3 4,3 4,4 4>} & $2\times10$ & 10 & 14.5\% & 1.4\% \\ \hline
4-star (Ref.~\cite{bartolucci2023fusion}) & \scalebox{1}[1.0]{\tt\scriptsize<2 3:2,1 2,1 2,2:6,2 4,4>} & $24$ & $3\times4$ & 6.9\% & 0.75\% \\ \hline
Pyrochlore & \scalebox{0.8}[1.0]{\tt\scriptsize<4 3:1 2 3 4,1 2 4,1 3 4,2 4:3 3 6,3 3,4>} & $2\times6+2\times18$ & $4\times6$ & 10\% & 0.78\% \\ \hline
Cuboctahedral & \scalebox{.6}[1.0]{\tt\scriptsize<8 3:1 2 3 4 5 6 7 8,1 2 3 4 5 7 8,2 4 6 8,3 5 8 7:3 3 3 3 3 4 3,2 4 4,4>} & $4\times3+12+24$ & $3\times8$ & 12.7\% & 1.3\% \\ \hline
\end{tabular}
\caption{Examples of surface-code fusion complexes, with their names, D-symbols, syndrome-graph check degrees C, resource-state sizes (number of qubits) R, and the erasure and Pauli thresholds $p^\star_\mathrm{erasure}$ and $p^\star_\mathrm{flip}$. The first two integers in the D-symbol denote the size and the dimension, respectively, where the dimension is always 3 in this work. See Ref.~\cite{gavrog, Delgado-Friedrichs2001, Delgado-Friedrichs2003} for the full description of D-symbols. Here the $n_1\times d_1 + \ldots$ notation means there are $n_1$ checks/resource-states with degree/size $d_1$, etc, in the unit cell. Note the XX and ZZ syndrome graphs for the Alternated-cubic fusion complex have different thresholds. 
}
\label{tab:fusion_complexes}
\end{table*}

\subsection{Threshold simulation} 

Once a fusion complex is identified \texttt{Gavrog} can provide a full cell-complex description in the form of boundary maps and a highly symmetric embedding of its vertices in 3d. With this information, we can specify the resource states, fusions, syndrome graphs, and logical membranes for the fusion complex as prescribed in Section~\ref{sec:homology}. We perform fault-tolerance simulations to compare fusion complexes.
We characterize the fault-tolerance capability of fusion complexes under two error models that capture erasure and Pauli errors.  

\begin{enumerate}
    \item Erasure model, where each fusion measurement outcome is erased with probability $p_\mathrm{erasure}$. We denote the threshold for erasure as $p_\mathrm{erasure}^\star$.
    \item Pauli model, where each fusion measurement outcome suffers a bit-flip error with probability $p_\mathrm{flip}$. We denote the threshold for pure Pauli error as $p_\mathrm{flip}^\star$.
\end{enumerate}   

Table~\ref{tab:fusion_complexes} shows the threshold results for the selected example fusion complexes discussed in the previous section under these two error models. Starting with the syndrome graph, we sample a random error configuration according to the probabilities $p_\mathrm{erasure}$ and $p_\mathrm{flip}$, compute the syndrome and apply decoding to identify a correction and then check for a logical failure. Since the structure of checks in fusion complexes can be represented by syndrome graphs we can apply standard decoders for surface codes, such as Minimum-Weight Perfect Matching (MWPM)~\cite{dennis2002topological,kolmogorov2009blossom,wu2023fusion,higgott2022pymatching} and Union Find (UF)~\cite{union_find}. We use a minimum weight perfect matching decoder and repeat this simulation over many trials and values of the error parameters to estimate the threshold. 

The performance of these schemes using linear-optics can be greatly increased using local encodings~\cite{bartolucci2023fusion,bell2023optimizing,hilaire2021error} and dynamic failure bases~\cite{bombin2023increasing}. Furthermore, these protocols can be used to balance the $X$ and $Z$ thresholds in cases where they would otherwise differ (as for the Alternated cubic complex, for example).

Performing threshold simulations for examples found in the \texttt{Gavrog} search, we find many fusion complexes with thresholds higher than those in Ref.~\cite{bartolucci2023fusion} (some are shown in Table~\ref{tab:fusion_complexes}). Typically, larger thresholds come from fusion complexes with larger resource states, although this relationship is not strictly monotonic. Identifying the most promising candidates requires analysis with more detailed error models of the hardware and system constraints, which is beyond this scope of this paper.




\section{Fault-tolerant complexes}
\label{sec:homology}

Our focus in Sections~\ref{sec:fusion_complexes} and~\ref{sec:automated_search} was on construction of fault-tolerant schemes for FBQC. But the fusion complex construction can be applied much more generally to fault-tolerant protocols made up of finite-sized operations. The vertices of the fusion complex, which in Section~\ref{sec:fusion_complexes} we identified with resource states, can instead by identified with other computational primitives such as the 2-qubit entangling gates relevant in circuit-based computation, non-destructive multi-qubit projective measurements used in floquet-based computation, or multi-qubit destructive measurements in generalized FBQC. We refer to these then as \emph{fault-tolerant complexes}. This picture allows us to extend homological descriptions of fault tolerance to a much wider class of protocols and error models than previously known, including circuit-level noise. It also enables the generative capabilities described in Section~\ref{sec:automated_search} to be applied to other hardware primitives. 

\subsection{Homological descriptions of fault tolerance}

Homological descriptions of surface-code fault tolerance (in space-time) were introduced in~\cite{raussendorf2007topological} in the setting of MBQC, where it was shown how all the fundamental objects of the system---qubits, checks, errors, and logical operators---could be represented through a 3-dimensional cell complex. The X-type and Z-type checks and logical correlators could be related to closed membranes (2d surfaces) in either the primal cell complex, or its dual complex. The checks are supported on (homologically) trivial closed surfaces and logical correlators are supported on (homologically) non-trivial closed surfaces. However a limitation of this model is that the elementary errors (such that the error model is an independent distribution over elementary errors) that can be represented homologically are restricted to single-qubit measurement errors, and imperfections arising from entanglement generation cannot be captured without correlations. Typically these errors are accounted for by computing the propagation of physical errors through circuits and introducing correlations to the error model. Alternatively the syndrome graph of the cubic graph can be modified by adding additional ``diagonal'' edges capturing the effect of correlated errors, but this new syndrome graph is not directly related to the original cell complex.An example of this is capturing the so-called `hook' errors~\cite{Dennis_2002} in syndrome extraction circuits by the addition of a `diagonal edge' to the cubic syndrome graph~\cite{wang2011surface}. The original property of duality between the X- and Z-type syndrome graphs is apparently lost. 

However we will show that starting from a fault-tolerant complex we are able to expand what can be captured as `elementary errors' allowing us to identify homological cell complexes representing the structure of X- and Z-syndrome graphs and logical membranes in the presence of noise during entanglement generation. They are not related by a simple duality, but rather by 6 related cell complexes, all of which can be derived from the fault-tolerant complex. 

For our current purposes it is the dual of the fault-tolerant complex that is most illuminating, we call this the \emph{resource homology complex} or \emph{R-homology complex}, as the cells correspond to the primitive operations (resources).

\subsection{The resource homology complex}

The resource homology complex (or R-homology complex) gives a homological representation of a fault tolerance protocol. We can first redefine the constraints in this picture as follows.

\begin{definition}\label{def:homologycomplex}
  A 3-dimensional cell complex, $\complexR$, is an R-homology complex if all faces are 4-sided.
\end{definition}

Each cell in the complex represents an elemental unit of the protocol (for example, a resource state in FBQC, or a 2-qubit entangling gate in CBQC), with one qubit on each face of the cell. The vertices are bi-colorable and the two colors correspond to X-type and Z-type check operators.

\begin{figure}[t]
    \centering
    \includegraphics[width=0.95\columnwidth]{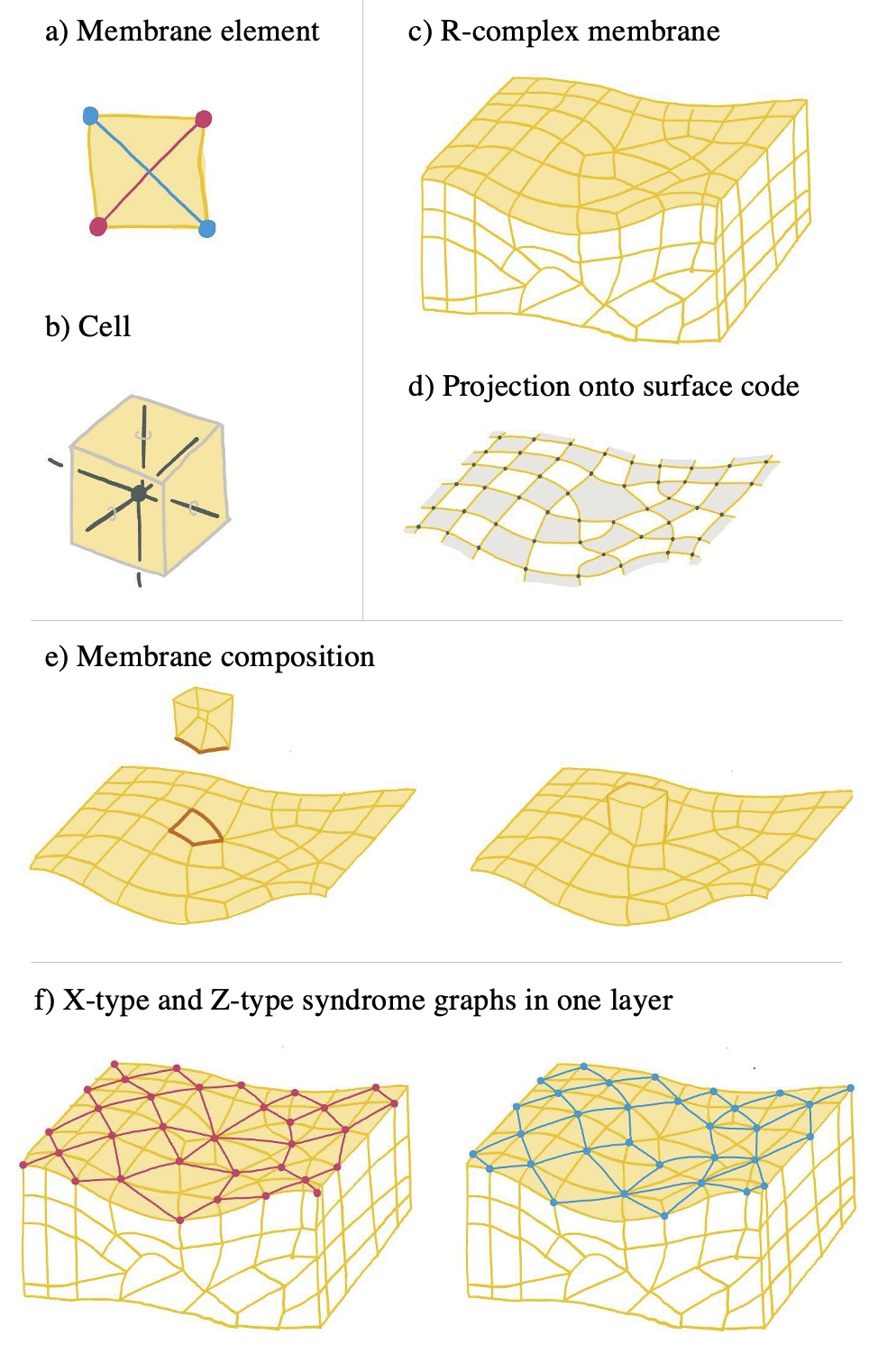}
    \caption{The resource homology complex (dual fusion complex). a) Every face of the complex has 4 edges. The vertices at the boundary are 2-colorable and correspond to check operators. X-type checks are shown as red circles and Z-type checks are shown as blue circles. b) Cells correspond to the elementary units of the fault tolerance protocol. In FBQC these cells are resource states with each face corresponding to a qubit. c) A region, $\region$, of the R-homology complex has a boundary $\partial \region$, which is a membrane made up of faces. When the protocol inside the region $\region$ has been performed what remains is a surface code (up to signs of the stabilizers) on the membrane. d) The dual of the 2-dimensional membrane gives the plaquette representation of the surface code with qubits on vertices and checks on 2-colorable faces. e) Two regions $\region$ and $\region'$ can be composed. $\region'$ represents some portion of a fault tolerance protocol. Once it is performed the new state is given by a surface code one the boundary $\partial (\region \cup \region')$ (up to stabilizer signs). f) Each membrane element contributes one edge to each the X- and Z-type syndrome graph. Within any membrane the X- and Z-type syndrome graphs are 2-dimensional and dual to one another. }
    \label{fig:R_homology_complex}
\end{figure}

\textbf{Membrane elements}.
The 4-sided faces (which we call square faces) of the R-homology complex are membrane elements, with one illustrated in Fig.~\ref{fig:R_homology_complex}a). Four check vertices surround each face, two X-type and two Z-type. Each face represents a distinct location at which an error can occur. A Pauli-X error on the face will cause the Z-checks on the boundary of the face to be flipped, and a Pauli-Z error on the face will cause the X-checks to be flipped. Consequently, each face has two edges of the syndrome graph associated with it, one X-type and one Z-type as indicated in the figure. The error could originate from either cell it is contained in but they both cause the same syndrome. 

\textbf{Membranes}.
The R-homology complex provides a formalization of our earlier description of surface codes existing on the boundaries of fusion networks. If we consider a region, $\region$, for which the fault tolerance protocol has been performed (e.g., the fusions (FBQC) or gates (CBQC) have been completed within $\region$), then the only remaining qubits are on the boundary of the region, $\partial \region$, which forms a membrane supporting a surface code. A membrane is composed of 4-sided faces (membrane elements), and is itself a 2-dimensional cell complex. If we take the 2d dual complex of the membrane surface, we arrive at the plaquette-style representation of a surface code~\cite{wen2003quantum}, where the vertices are qubits and the faces are bi-colorable and represent the X- and Z-stabilizers of the code.

\textbf{Cells}.
The 3-cells of the R-homology complex represent the smallest elements of a fault tolerance protocol (see examples in Sec.~\ref{sec:example_FT_complexes}). The boundary of a cell is a closed membrane, such that we can always interpret the cell boundaries as a surface code on a topological sphere. 
The cells define the elementary error model, determining the errors that can captured without correlations.
Elementary errors can be captured by an error model in which the protocol inside the cell is performed perfectly, but followed by independent errors on its outputs. If we choose cells that are very large the elementary error model is likely unphysical as it would assume a large entangling operation on many qubits that produces uncorrelated errors on the outputs. Any more structure in the errors would need to be accounted for with correlated errors, which is no longer information captured in the geometric structure of the complex. To capture physical level noise we should consider small cells where the noise model of the primitives can be represented with independent errors on their outputs. For example, one can consider contracting or splitting cells along faces to obtain bigger/smaller cells.

\textbf{Composition of membranes}.
Fig.~\ref{fig:R_homology_complex}e) illustrates the composition of membranes. We consider starting with some region, $\region$, of the cell complex, $\complexR$. $\region$ represents the portion of the protocol that has been performed. The boundary of the region $\partial \region$ corresponds to the state of the qubits. Then we perform an additional protocol element associated with a cell, $c$, expanding the region to $\region \cup c$ (e.g., in FBQC, the protocol element corresponds to preparing an additional resource state on $c$ and fusing it along one or more faces). Then the boundary of the new region, $\partial (\region \cup c)$, corresponds to the updated quantum state. For two cells to connect, their faces must have the same orientation such the color of their vertices matches.

\textbf{Errors and syndrome graphs}.
Each face represents a distinct error location, and an error associated with a given face causes two check vertices at the boundary of the face to be flipped. The \emph{syndrome graph} is a representation of the check structure in surface codes where there is a vertex for each check and an edge for each potential error. There is an X-syndrome graph representing X-checks and Z-errors, and similarly a Z-syndrome graph representing Z-checks and X-errors. Each face in $\complexR$ contributes two edges to the syndrome graph, one for each error type, as shown in Fig.~\ref{fig:R_homology_complex}a). Now consider combining these syndrome graph elements across a membrane of the R-homology complex. An example is shown in Fig.~\ref{fig:R_homology_complex}f). The result is the syndrome graph of the surface code supported on the membrane. Consequently, the X- and Z-syndrome graphs within this layer are dual to each other in the sense that they can be specified as the vertices and edges (X) or dual vertices and dual edges (Z) of the same 2d cell complex.
This property of 2-dimensional duality of syndrome graphs is the same for every surface in the R-homology complex. The full syndrome graph is built up by combining the graph elements associated with every face. The final structure is 3-dimensional, but we can now see that the relationship between X- and Z-type components is derived from a 2-dimensional duality which elucidates why when circuit level noise is included in fault tolerance protocols the full X-type and Z-type syndrome graphs cannot be considered as directly dual to one another in 3-dimensions.

\subsection{Examples of fault-tolerant complexes}\label{sec:example_FT_complexes}

The cells of R-homology complexes (vertices of the fault-tolerant complex) represent primitive quantum operations. The ZX formalism~\cite{coecke2011interacting, backens2014zx,coecke2018picturing,ZXhappyfamilies} provides a helpful tool for identifying how these constructions can translate between different types of primitives. In the ZX formalism any quantum operation can be constructed as a ZX-diagram (e.g., those shown in Fig.~\ref{fig:ZX_resources}), and from there it can be restated as a surface-code cell (e.g., those shown in Fig.~\ref{fig:resource_states}). This identifies the structure of the primitive cells of the R-homology complex.

\textbf{Example: FBQC with 2 qubit fusions.} In the FBQC setting where all fusions are 2-qubit Bell basis fusions the cells are identified with resource states and fusions are performed between every pair of neighboring cells. This is the case covered in detail in Section~\ref{sec:fusion_complexes}.

\textbf{Example: FBQC with multi-qubit fusions.} With multi-qubit projective measurements each cell of the R-homology complex is identified with either a resource state or a fusion measurement. The cells must be bi-colorable such that fusion cells only share faces with resource-state cells and vice versa. 

\textbf{Example: CBQC with 2-qubit entangling gates.} The cells have 4 faces representing the two inputs and two outputs of 2-qubit entangling gates. The geometry of the cells will depend on the gates and can be determined from their ZX diagrams. 

\textbf{Example: Floquet-based fault tolerance~\cite{hastings2021dynamically, paetznick2023performance, haah2022boundaries, gidney2022pair, kesselring2022anyon, davydova2023floquet, ZXhappyfamilies, paesani2022high, aasen2022adiabatic, townsend2023floquetifying, dua2023engineering}.} Primitives are multi-qubit projective measurements. For example a 2-qubit $XX$ projection would correspond to a cell with 4 faces corresponding to the two inputs and two outputs of the measurement. This cell would be equivalent to a 4-qubit GHZ-type resource state. 

The fault-tolerant complex construction is limited to describing protocols that behave as surface codes down to their basic operations. There are circuits that implement surface codes at a higher level, but in such a way that at a lower level some other structure is introduced. For example in CBQC it is possible to perform syndrome extraction such that some errors cause 4 checks to be violated (e.g.,~\cite{geher2023tangling}). These cases cannot be fully captured by a surface-code fault-tolerant complex. Rather, we would expect there to be some length scale at which the protocol can be partitioned into a fault-tolerant complex, whilst within those cells the lower level error structure would need to be described with correlations. Gaining a better understanding of circuits that go beyond a surface code structure at a low level is a topic for future study.

\subsection{Homology of checks and logical operators}\label{sec:triplets}

\begin{figure*}
    \centering
    \includegraphics[width=2\columnwidth]{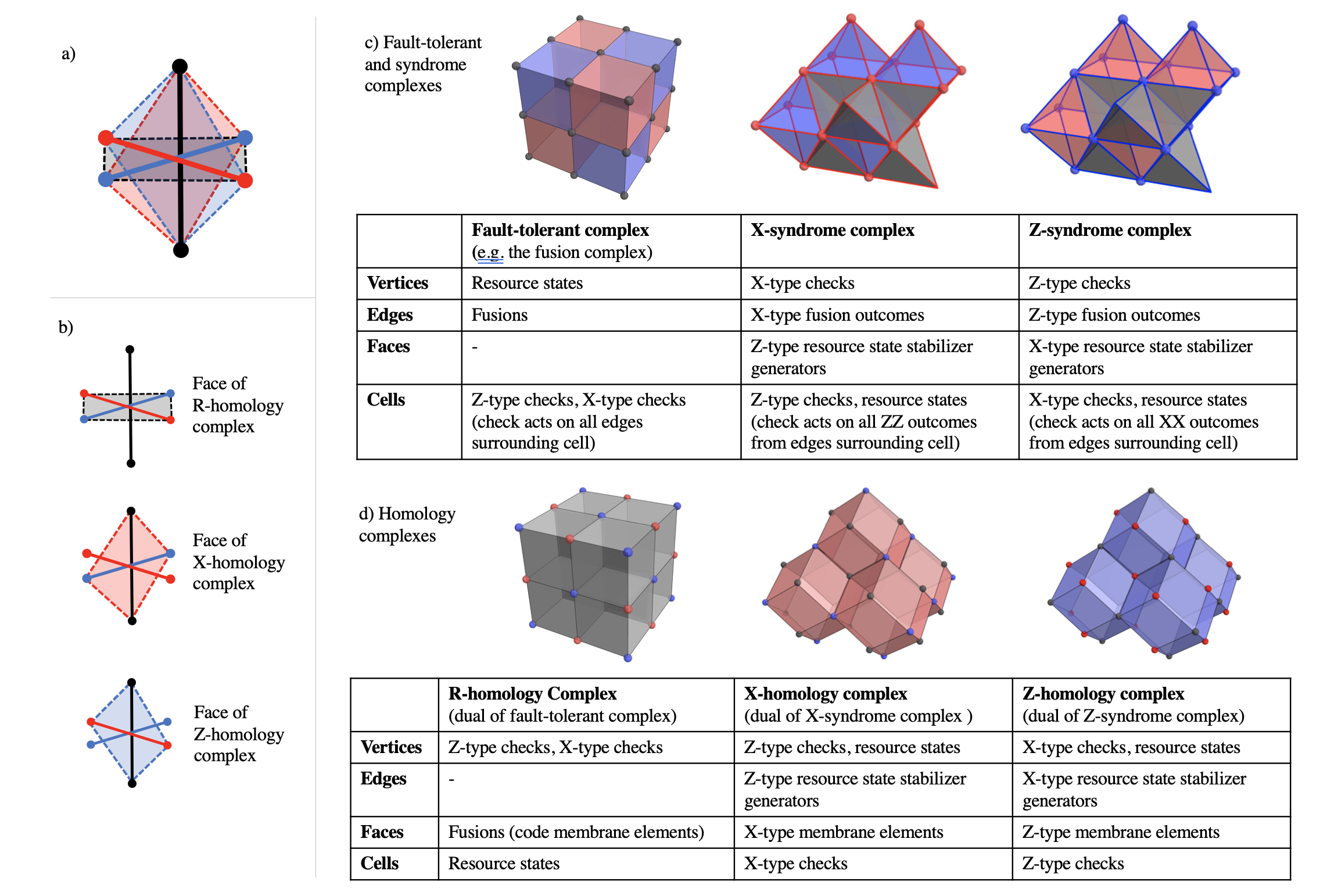}
    \caption{a) Triplet relationship between the R-homology complex (dual fusion complex), X-homology complex and Z-homology complex. Faces from each complex intersect to form an octahedron. b) Face of each homology complex shown separately. Edges of the fusion complex and X- and Z-syndrome complexes are shown. Tables c), d) show the example of the cubic fault-tolerant complex and the syndrome complexes and homology complexes that derive from it.}
    \label{fig:triplet}
\end{figure*}

Having identified the origin of the X- and Z-syndrome graph structure for a fusion complex, we now ask if there is another cell complex for which the faces represent membrane elements of X- (or Z-) logical correlators? In such a complex, cells should represent check operators and membranes would be used to construct logical correlations of the fault-tolerant protocol. We recall that in the case of measurement-based fault tolerance~\cite{raussendorf2005long,raussendorf2007topological} there is a large cluster state defined by a cell complex such that the faces represent membrane elements of X- logical correlator and the faces of the dual cell complex represent membrane elements of the Z- logical correlator. In our present case, where the primitive operations are finite in size, the relationship is not so straightforward.

However we find that we can identify these structures, which we will call the X-homology complex and the Z-homology complex. These two new complexes along with the R-homology complex form a triplet, whose relationship is visualized in Fig.~\ref{fig:triplet}a) and b). Each face in one complex has a corresponding face in each of the other complexes and the three intersect to form an octahedron. Each complex can be derived by taking a different cross-section through the triplet octahedron. This means that the X- and Z-homology complexes are, perhaps surprisingly, subject to the same constraints as the R-homology complex; the faces are squares, and by extension the vertices are bi-colorable.

In the X-homology complex the faces correspond to X-membrane elements. The cells, which are the smallest closed membranes, correspond to X-type check operators. Logical operators are formed from extensive membranes, which are 2d sheets of faces. In the X-homology complex the two colors of vertices correspond to Z-type checks and resource states. Similarly in the Z-homology complex, the faces correspond to Z-membrane elements, the cells are Z-type check operators, and the two colors of vertices correspond to X-type checks and resource states. Notice the similarities and differences here to the construction in the fault-tolerant MBQC scheme of Ref.~\cite{raussendorf2007topological}, where closed 2d surfaces in the primal and dual complex correspond to checks and logical membranes. The membranes behave identically, but when it comes to the edges and vertices the similarity is lost, and moreover the X- and Z-homology complexes here are not dual to one another.

We can gain further insight by inspecting the dual of the homology complexes. We call these `syndrome complexes' since, by construction, the 1-skeleton of the dual of the X-homology complex is the X-type syndrome graph. These syndrome complexes have the same general properties as fusion complexes: that they must have four faces incident at each edge. The cells of the complex, which are bi-colorable, represent either resource states or checks. The faces of the syndrome complex represent the resource-state (or other primitive operation) stabilizer generators. These are the `trivial’ non-detectable errors corresponding to the loop in the syndrome graph that is the boundary of that face. Note that the size of these faces is related to the size of physical primitives available, and gives some insight into why syndrome graphs that account for circuit-level errors often appear to be triangulated.

\textbf{Interchanging fault-tolerant and syndrome complexes}.
Because the fusion complex (or in general, fault-tolerant complex) and syndrome complexes all obey the same constraints they are interchangeable in the following sense: we are free to identify any one complex as the fusion complex, and it follows that the other two represent the X- and Z-syndrome complexes. This framework gives another way of constructing fusion networks by starting from a target syndrome graph. For example, suppose we want to construct a fusion network that has a cubic syndrome graph. We can start with the cubic cell complex as the X-syndrome complex, compute the rest of the triplet and identify that the fusion complex is the ``alternated cubic'' complex. This example is simply a permutation of the labels in the example of Fig.~\ref{fig:triplet}. Different fusion complexes from the same triplet will generally have distinct geometries, and therefore will not have the same threshold. 

\subsection{Comparing ZX-diagrams with fault-tolerant complexes}

Fault-tolerant complexes provide a general, geometric description of fault tolerance protocols. Recently Ref.~\cite{ZXhappyfamilies} also described how the graphical language of the ZX-calculus can be used as a unifying framework for different computational models for fault tolerance. In Figs~\ref{fig:resource_states} and~\ref{fig:ZX_resources} we show the fault-tolerant complex and ZX-diagram representations of the example resource states in Section~\ref{sec:automated_search}.  We briefly comment on the complementary roles of these two approaches. The ZX-diagram language provides a clear framework for identifying commonalities between protocols when the primitive operations may appear very different (such as between circuit models and floquet-based schemes). By mapping operations to their underlying structure in the ZX-calculus the protocols can be re-expressed in a common language. This framework is very useful for translating methods between hardware implementations. However, this framework does not (as far as we can tell) enable a constructive approach to fault tolerance, schemes must be known a-priori and can then be reconstructed in the ZX formalism. Also, while the ZX-diagrams provide a complete description of the operations, they don't contain a description of the check operators which must be separately identified. Fault-tolerant complexes, on the other hand, contain a definition of both the operations and the check structure providing a framework to study the relationship of the two. Unlike ZX-diagrams fault-tolerant complexes also allow a generative search, and other constructive methods (see Appendix~\ref{app:construction}), so they can be used as an approach to identifying new schemes. The two frameworks play complementary roles in the study of fault tolerance schemes. In particular the ZX-diagram language provides a natural route to map the fault-tolerant complex construction into other computational models, as described in Sec.~\ref{sec:example_FT_complexes}. An appropriately designed filter-search could then identify candidate complex structures composed out of the required primitives. Finally we can map back to ZX-diagrammatic language to identify suitable qubit-worldlines and time ordering to construct a valid protocol.

\subsection{Boundaries and topological features}
We have described fault-tolerant complexes as they make up the bulk of a fault-tolerant protocol. In order to perform logical operations, we can insert topological features---such as boundaries, symmetry defects, and twists---into the complex in certain configurations~\cite{bombin2023logical}. In the FBQC setting these features can be realized by modifying the fusion measurements only, and do not require any alterations to the resource states. Examples of how to do this with the 6-ring network are shown in Section~VII. of Ref.~\cite{bombin2023logical}, and we describe how to implement them in a general fusion complex in Appendix~\ref{app:features}.




\section{Discussion}

We have introduced the fusion complex framework to geometrically define a large family of fault-tolerance protocols for many schemes of computation, which is particularly useful for FBQC. This construction allows a single geometric representation of fault-tolerance which captures the physical operations (resource states and fusions) as well as the check structure (syndrome graph). The fusion complex construction enables an automated search and simulation of a large number of fusion networks, which we have used to generate 627 new examples. Several of these complexes have an  improved threshold performance against both erasures and Pauli errors compared to previously publishes fusion networks. The examples we present here are based on uniform tilings, but there is no requirement for a uniform, periodic structure in the fusion complex, nor a requirement for the complex to be specified in advance of running the protocol, potentially being chosen dynamically at runtime.
 
Similar methods could be applied to find fault-tolerant complexes for other computational schemes (e.g., CBQC, FloBQC), by designing filter rules for complexes to constrain the primitives to those physically implementable in a given hardware paradigm. We hope that this will enable an acceleration of the search for improved fault-tolerance protocols for other computational models. 

The relationship between R-homology complexes and fusion complexes reveals a fundamental, geometric, relationship between the size of the computational primitives and the check weight and check structure of the protocol. As the development of fault tolerance methods becomes increasingly entwined with physically derived error models we hope this extension of geometric relationships will provide a useful theoretical tool for improving error tolerance. Future work may aim to solidify these relationships, possibly placing bounds on achievable thresholds given a particular set of primitives.

We have covered in detail only the bulk fault-tolerance properties, but a fusion complex representation can be a valuable tool for understanding the behavior of boundaries and other topological features---for instance allowing us to explore what modification of fusions (or resource states) can be used to introduce defects. In Appendix~\ref{app:features} we briefly outline the fusion complex representation of various topological features, and fully exploring the space of schemes remains an area for further study.

Beyond the surface code complexes described here, the same principles can be extended to other topological codes. In Appendix~\ref{app:color_code} we give related cell complex constructions of color codes, along with new examples. In Appendix~\ref{app:subsystem_code} we present related cell complex constructions for 3d single-shot subsystem codes suitable to some CBQC architectures, and we provide several new examples. The framework can be similarly used to carry out a generative search for new schemes. We hope future work can study the detailed performance of some of the new examples identified. In addition it would be interesting to consider schemes beyond topological codes. Desirable examples to study may include quantum LDPC codes~\cite{bravyi2014homological,tillich2013quantum,breuckmann2021quantum}, which when concatenated with topological codes can produce schemes with high thresholds and improved encoding efficiency~\cite{bombin2023logical,pattison2023hierarchical}.




\section*{Author contributions}
NN: Identified the fusion complex construction. 
NN, SR, HB, FP, YL: Developed framework, definitions and related fault-tolerant cell complex constructions including extensions in appendices.
YL: Developed search and filter methodology using Gavrog and performed simulations computing thresholds. 
CD, YL: developed cell complex based simulation framework
MP, TF, SR, YL: Studied fault tolerance properties of non-cubic fusion complexes. 
NN, SR, YL: Wrote the manuscript with input and review from the other authors. HB was on leave during the writing of the manuscript.

\section*{Acknowledgements}
The authors would like to thank Terry Rudolph,  Daniel Litinski and Dave Tuckett for many valuable conversations and their input on the manuscript, and all of our colleagues at PsiQuantum for helpful discussions.



\appendix

\section{Different representations of 2d surface codes}\label{app:surface_codes}

There are two common representations of 2d surface code families. We refer to these as the Kitaev and plaquette-style versions, the latter being introduced by Wen~\cite{wen2003quantum}.
The two representations are equivalent, in that any Kitaev surface code can be represented as a plaquette surface code (and vice-versa). Each is useful in different contexts. Indeed, the two are related by the Medial graph construction~\cite{bombin2007optimal,anderson2013homological}.

In particular, the definition of surface codes in Def.~\ref{def:surface_code} is the plaquette version, where qubits are on vertices of a 4-valent graph, and checks correspond to each plaquette (X-type or Z-type according to the face coloring).

The Kitaev surface code, on the other hand, can be defined on any 2d cell complex, but placing qubits on edges, and identifying vertices with X-checks and faces with Z-checks, as follows.

\begin{definition}\label{def:kitaev_surface_code}(Kitaev surface code)
  A stabilizer code $\mathcal{S}$ is a Kitaev surface code, if there exists a 2d cell complex $\complex_K = \{F, E, V\}$, with a qubit on each edge $e\in E$, such that $\mathcal{S} = \langle X_v, Z_f ~|~ \forall v\in V, f \in F \rangle$, where $X_v$ is a product of Pauli-$X$ operators on all edges incident to the vertex $v$, and $Z_f$ is a product of Pauli-$Z$ operators on all edges belonging to the face $f$. 
\end{definition}

Mapping from the Kitaev complex $\complex_K$ to the plaquette complex $\complex_P$ is achieved by the medial graph construction~ \cite{bombin2007optimal,anderson2013homological}. For every edge $\complex_K$ we place a vertex in $\complex_P$. We place an edge between two vertices in $\complex_P$ whenever the corresponding edges in $\complex_K$ are neighbours belonging to the same face.

\section{Constructive method for fusion networks based on intersecting spheres}
\label{app:construction}
Another way to generate fusion networks is based on intersecting spheres. This is motivated by an observation in 2d, that a valid 2d surface code can be generated by any configuration of intersecting circles (or lines, in general), where no more than two circles are allowed to intersect at a point. This defines a valid surface code with qubits as vertices (where circles intersect), and each closed area defining a check. The checks are bi-colorable, and each color can be assigned as Z-type or X-type arbitrarily.

One can similarly define a valid surface-code fusion complex by any configuration of intersecting spheres (or planes), where no more than two spheres may intersect along a line and no more than three may intersect at a point. This defines a cell complex with the correct properties for a fusion complex, as per Def.~\ref{def:fusion_complex}. An example is shown in Fig.~\ref{fig:sphere_construction}.

Unlike the search and filter method outlined in Section~\ref{sec:automated_search}, this is a constructive method and so provides a different path to identifying new instances of interest.

\begin{figure}
    \centering
    \includegraphics[width=\columnwidth]{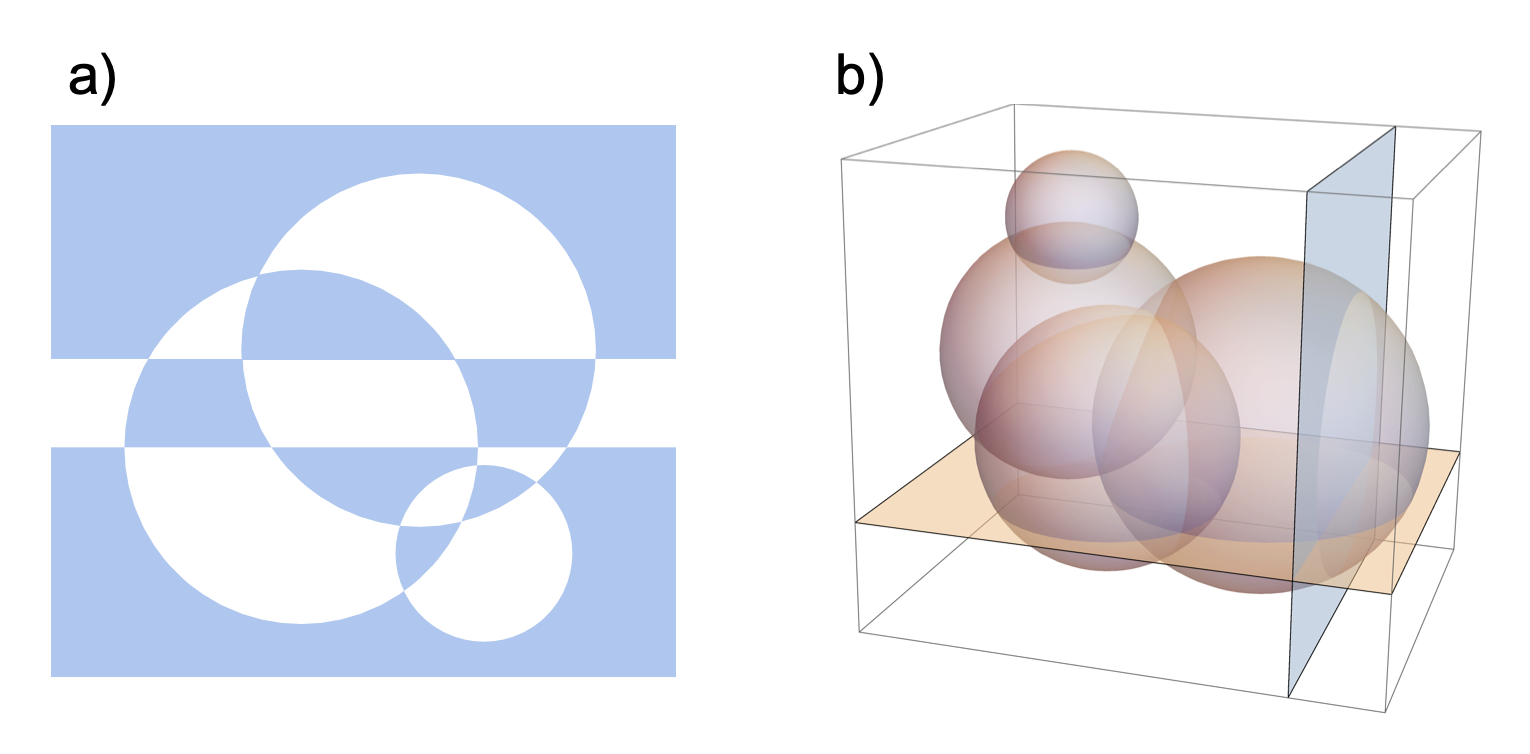}
    \caption{Illustration of construction of surface codes and fusion complexes from intersecting circles or spheres. a) A 2d surface code is constructed from intersecting circles, as long as no more than 2 circles intersect at a single point. Each intersection is a qubit and each closed area is a check. The surface is bi-colorable by construction. b) A fusion complex can be formed of intersecting spheres in 3-dimensions. No more than three spheres may intersect at a single point, and no more than two spheres may intersect along a line. The three cells created are bicolorable by construction, and each edge has four incident faces. Any slice through this will be a surface code. 
}
    \label{fig:sphere_construction}
\end{figure}

\section{Topological features in fusion complexes}\label{app:features}
In this section we show how to insert topological features into a fusion complex by modifying the measurement pattern. The types of features we consider are boundaries (both X-type and Z-type), the domain wall (also known as the $\zz_2$ domain wall), and twist defects, which are sufficient to perform any fault-tolerant quantum computation with surface codes~\cite{raussendorf2007topological, raussendorf2007fault, bombin2009quantum, bombin2010topological, horsman2012surface, fowler2012time, barkeshli2013twist, hastings2014reduced, yoder2017surface, brown2017poking,litinski2019game,landahl2021logical,bravyi2005universal,bravyi2012magic,bombin2023modular,litinski2019magic,li2015magic,bombin2022fault,chamberland2022circuit}. 

\textbf{Boundaries}.
The surface code has two types of 2d boundaries (here, boundary refers to a ``gapped boundary to vacuum")~\cite{kitaev2012models, levin2012braiding, lan2015gapped}, which we call X-type and Z-type (they are also known as rough and smooth in the literature~\cite{bravyi1998quantum}). One can understand the boundaries in terms of what types of checks are supported there. To define a X-type or Z-type boundary, we specify a surface in the dual complex (the R-homology complex), forming the boundary of a region of the fusion complex. Each face in this surface corresponds to a qubit from a resource state. To prepare an X-type (Z-type) boundary, we measure all qubits supported on this surface in the single-qubit X-basis (Z-basis). In doing so, only the vertices of the dual complex of one color support checks, namely the X-type (Z-type) checks. 

Another way of understanding these boundaries is as follows. Consider the surface-code surface state on the boundary after all fusions are performed. Then inserting an X-type (Z-type) boundary is the same as performing a transversal readout in the X-basis (Z-basis).

\textbf{Domain walls}.
A domain wall corresponds to a local breaking of the bi-colorability of the lattice. Consider the fusion complex. A domain wall is specified by 2d surface in the primal complex. We first consider a surface without boundary. Each face of this surface is incident to two 3-cells of different type (X-type and Z-type). A domain wall corresponds to a modification of the cell complex, whereby pairs of 3-cells on either side of the surface are ``glued'' together. Gluing two 3-cells together along a face gives a single 3-cell where the gluing face is removed (meaning the face, along with its edges and vertices are removed). Edges sharing a vertex on the gluing face become a single edge defined by their two distinct end points. 

This new cell complex corresponds to a modified fusion complex. Removed vertices and edges along the gluing surface mean some resource states and fusion are removed, respectively. Fusions along the glued edges are Hadamarded Bell measurements $\{XZ, ZX\}$. Correspondingly, the glued 3-cells correspond to mixed X- and Z-type checks. On either side of the domain wall, their restriction agrees with that of the usual X- or Z-type check in the bulk fusion complex. See Section~VII. of Ref.~\cite{bombin2023logical} for examples in the 6-ring network.

\textbf{Twist defects}.
Twist defects arise when we consider domain wall surfaces with boundaries~\cite{bombin2010topological}. They are 1d objects that live on the boundary of domain walls. Consider a domain wall surface with boundary, consisting of a set of edges of the fusion complex which we call the twist edges. On each of these edges, there are four incident faces (by definition of the fusion complex). One of these faces must belong to the domain wall surface--which we call the twist faces--while the others are usual bulk faces. We modify the fusion complex around the twist faces as follows. Each pair of 3-cells incident to a twist face is partially glued: we remove the twist face, and all edges and vertices on the twist face apart from the twist edge (and its corresponding vertices). 

The new cell complex corresponds to a modified fusion complex as follows. Resource states and fusions in the bulk of the domain wall are modified as per the previous domain wall rules. Resource states and fusions along the twist edges remain as per the bulk rules. There will be one unpaired qubit per resource state on the vertices of the twist edges, as the resource state they would otherwise be fused to is supported on the domain wall and thus removed. These qubits are to be measured in the single qubit $Y$ basis. The corresponding check operators along these partially glued 3-cells is given by the product of \textit{(i)} $XX$ or $ZZ$ measurements on edges of the 3-cell, determined by the cell-type on their respective side of the domain wall, \textit{(ii)} $XZ$ or $ZX$ for edges crossing the domain wall (again according to the cell type on either side of the domain wall), and \textit{(iii)} $YY$ or $Y$ for edges within the domain wall.

\section{Higher dimensional fusion networks}\label{app:higherdim}
One can extend the fusion-complex construction to higher dimensions using GHZ states as resource states (recall, GHZ states are a type of non-degenerate surface code state). This construction is motivated by the homological viewpoint of fault-tolerant cluster states~\cite{raussendorf2007topological,nickerson2018measurement}. 

Consider a $D$-dimensional cell complex $\mathcal{L}_D = \{C_D, C_{D-1}, \ldots, C_1, C_0\}$, where each $C_i$ corresponds to a set of $i$-dimensional cells. One can construct a GHZ-based fusion network as follows. Choose a dimension $k\in \{1,\ldots, D{-}1\}$. Then on each $k$-cell $c_k$, we place an $n_k$-qubit GHZ resource state, where $n_k$ is the number of $(k{+}1)$-dimensional cells incident to $c_k$. Similarly, on each $(k{+}1)$-dimensional cell $c_{k+1}$, we place an $n_{k+1}$-qubit GHX resource state, where $n_{k+1}$ is the number of $k$-dimensional cells belonging to $c_{k+1}$. Here, a GHX state is a GHZ state where a Hadamard has been applied to each qubit. Fusions are performed between every pair of resource states on a $(k{+}1)$-cell $c_{k+1}$ and a $k$-cell $c_k$ whenever $c_k$ belongs to $c_{k+1}$. We get an X-check for every $(k{-}1)$-dimensional cell $c_{k-1}$, given by the product of all $XX$ outcomes on fusions performed on resource supported on cells incident to  $c_{k-1}$. We get a Z-check for every $(k{+}2)$-dimensional cell $c_{k+2}$ given by the product of all $ZZ$ outcomes on fusions performed on resource supported on cells belonging to $c_{k+2}$.

When $D=3$, the corresponding fusion networks in this construction are fusion complexes in disguise; one can define a fusion complex giving rise to the same resource states and fusions (and therefore checks).

\section{From fusion complexes to 3d subsystem codes}\label{app:subsystem_code}

The fusion complex framework can also be used to generate 3d subsystem codes useful for single-shot quantum error correction~\cite{bombin2015single}. We provide a mapping from fusion complexes with bi-colorable vertices to 3d subsystem codes, and present some new examples. These subsystem codes are qubit subsystem-code versions of the more traditional 3d stabilizer toric codes. For example, the 3d cubic lattice fusion complex maps to the 3d subsystem toric code of Ref.~\cite{kubica2022single}. Thus, the fusion complex framework gives a constructive and generative way to generalize the subsystem toric code (STC) models of Refs.~\cite{kubica2022single,bridgeman2023lifting}, providing a family of examples of 3d single-shot error correcting codes that have been of recent interest~\cite{bombin2015single,brown2016fault,kubica2022single,bridgeman2023lifting, stahl2023single}. 

A subsystem code \cite{poulin2005stabilizer,bacon2006operator} is specified by a (not-necessarily commuting) subgroup of Pauli operators $\gauge$ called the gauge group. Stabilizers for the code are obtained (up to phases) by $\stabilizer = \mathcal{Z}(\gauge) \cap \gauge$, where $\mathcal{Z}(\gauge)$ consists of all Pauli operators commuting with all elements of $\gauge$. Thus in particular, $\stabilizer \subset \gauge$, such that stabilizer outcomes can be inferred by measuring gauge generators. Some subsystem codes, like gauge color codes and subsystem toric codes, allow ``gauge fixing" between different stabilizer codes (meaning the target stabilizer code is described by a subgroup of the gauge group)~\cite{paetznick2013universal,bombin2015gauge}. Here we show that any fusion complex with bi-colorable vertices can be used to define a subsystem toric code.

\begin{definition}\label{def:subsystem_code}(3d subsystem-toric-code complex)
A 3-dimensional cell complex, $\complexF=\{C,F,E,V\}$, is a subsystem-toric-code complex if every edge has exactly four incident faces, and it has bi-colorable vertices. 
\end{definition}
The subsystem code is defined on such a complex as follows (recall that the 3-cells are also bicolorable by construction).
\begin{itemize}
    \item A qubit is associated to each 1-cell $e\in E$.
    \item A gauge generator $g\in \gauge$ is associated to each corner of a 3-cell, and is $X$-type or $Z$-type according to the color of the 3-cell. 
    \item A stabilizer $s\in \stabilizer$ is associated to each 3-cell $c \in C$, and is $X$-type or $Z$-type according to the color of the 3-cell.
\end{itemize}

In particular, the stabilizers for each 3-cell $c\in C$ are a product of either $X$ or $Z$ (depending on the color of $c$) on all qubits supported on edges of that cell. To define the gauge generators, take an $X$-type ($Z$-type) 3-cell $c\in C$, and a vertex $v$ supported on $c$. Then the gauge generator $g_{(v,c)}\in \gauge$ associated to the corner $(v,c)$ of the 3-cell is given by a product of $X$ ($Z$) operators on all qubits supported on the edges of $c$ that are incident to $v$. Stabilizers commute with each other and all gauge generators. 

Mapping from the FBQC version of the fusion complex to the 3d subsystem code, one simply replaces fusions by qubits, and maps the support of resource-state stabilizers and check operators to these qubits appropriately: resource state stabilizers correspond to gauge generators and checks of the fusion network correspond to stabilizers of the subsystem code. Thus the set of gauge generators around each vertex defines a 2d stabilizer surface code on a sphere. While all 2d surface code stabilizers commute within each surface code on each vertex, neighbouring surface codes overlap, and their stabilizers do not necessarily commute.

We can verify that this construction gives a well-defined subsystem code with $\stabilizer = \mathcal{Z}(\gauge) \cap \gauge$ (up to signs). To do so, we show that (i), gauge generators commute with stabilizers, and (ii) that $\stabilizer \subset \gauge$. For (i), all gauge generators around a vertex form a 2d surface code, and thus commute. Gauge generators therefore commute with stabilizers, as each gauge generator can be obtained by restricting a stabilizer on a cell $c\in C$ to one of its corners. For (ii) we use the bi-colorability of the vertices of the complex. In particular, the corners of the 3-cell $c \in C$ supporting a stabilizer, are partitioned into two colors. Taking the product of gauge generators on the corners of (either) one of the colors gives the stabilizer on that 3-cell. Note that each stabilizer can be obtained in two ways as a product of gauge generators. This redundancy also enables single-shot error correction for these models (although we do not prove this here). 

The cell complex has bi-colourable vertices if and only if its 1-skeleton has even cycles (when viewed as a graph). In other words, each face of the complex must have even length. This rules out any complex with triangular faces (and thus tetrahedral volumes). Such complexes do not form valid subsystem surface codes.

Noting that the FBQC realization of the fusion complex is also described by a (different kind of) subsystem code (see Refs.~\cite{bartolucci2023fusion, bombin2023logical}), where fault-tolerant computation is described as a gauge fixing from the resource group to the fusion group. That single-shot 3d codes and (2+1)d fault-tolerant quantum computation are described by the same underlying fusion complex may clarify the connection between fault-tolerant quantum instruments and single-shot codes. 

\textbf{Example: cubic complex}.
The cubic complex corresponds to the 3d STC of Ref.~\cite{kubica2022single}. Gauge generators are all 3-body operators, and stabilizers are 12-body operators.

\begin{figure}
    \centering
    \includegraphics[width=0.3\columnwidth]{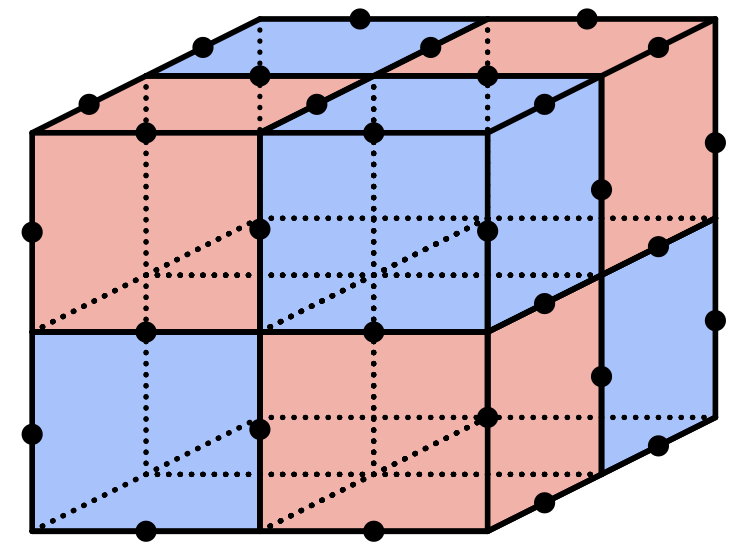}
    \includegraphics[width=0.3\columnwidth]{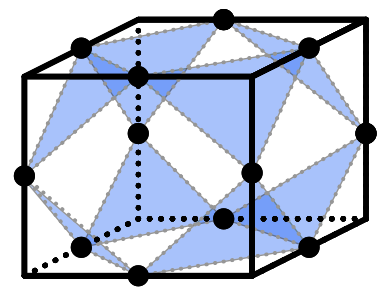}
    \includegraphics[width=0.3\columnwidth]{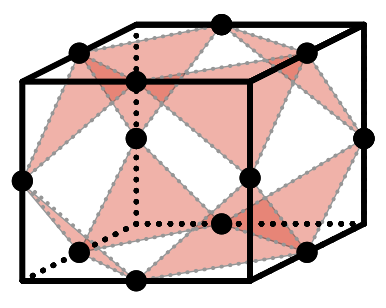}
    \includegraphics[width=0.3\columnwidth]{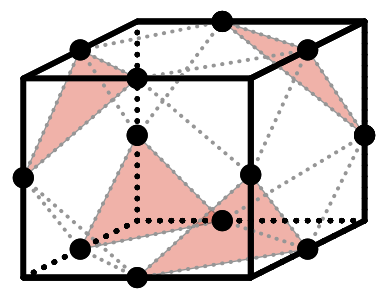}
    \includegraphics[width=0.3\columnwidth]{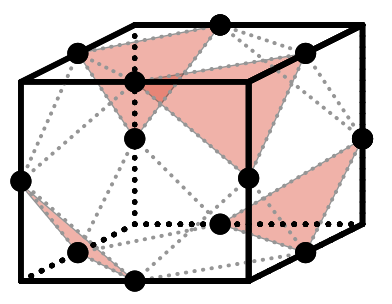}
    \caption{(top left) Cubic fusion complex with qubit locations on edges. (top middle, right) $X$ and $Z$ gauge generators on each corner, in blue and red, respectively. (bottom) Two ways of obtaining a $Z$ stabilizer for each $Z$-cell (similarly for $X$). }
    \label{fig:cubic_subsystem_code}
\end{figure}

\textbf{Example: 4-star complex}.
The 4-star complex corresponds to a 3d STC with 4-body and 2-body gauge generators, and stabilizers that are 24-body operators. 

\begin{figure}
    \centering
    \includegraphics[width=0.8\columnwidth]{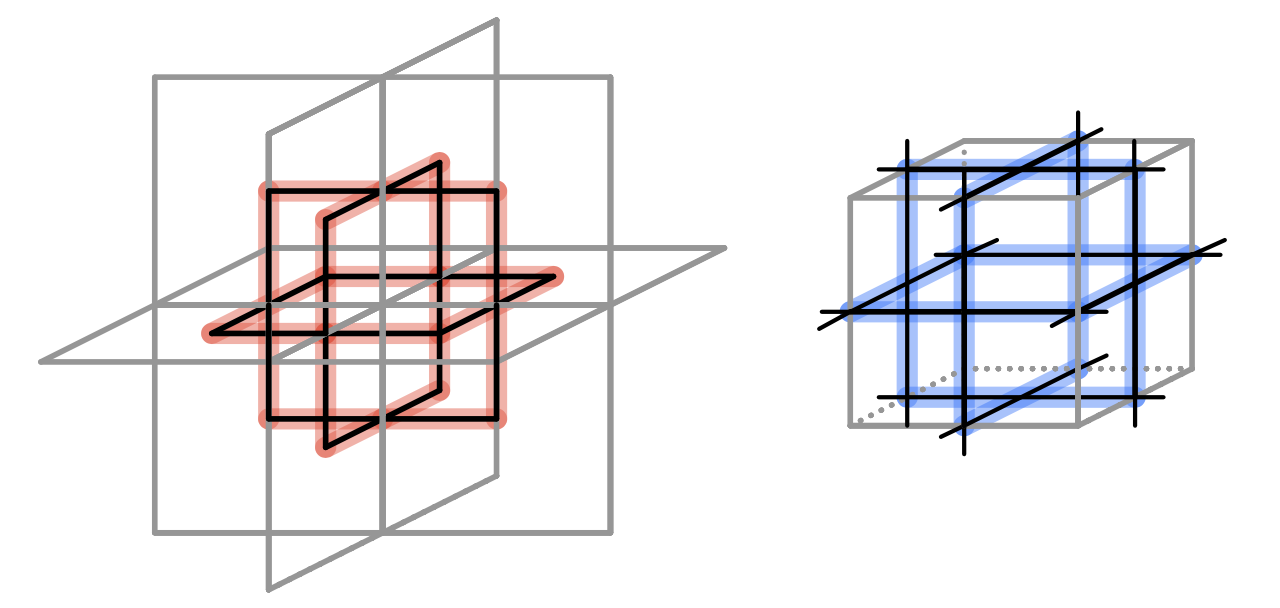}
    \includegraphics[width=0.95\columnwidth]{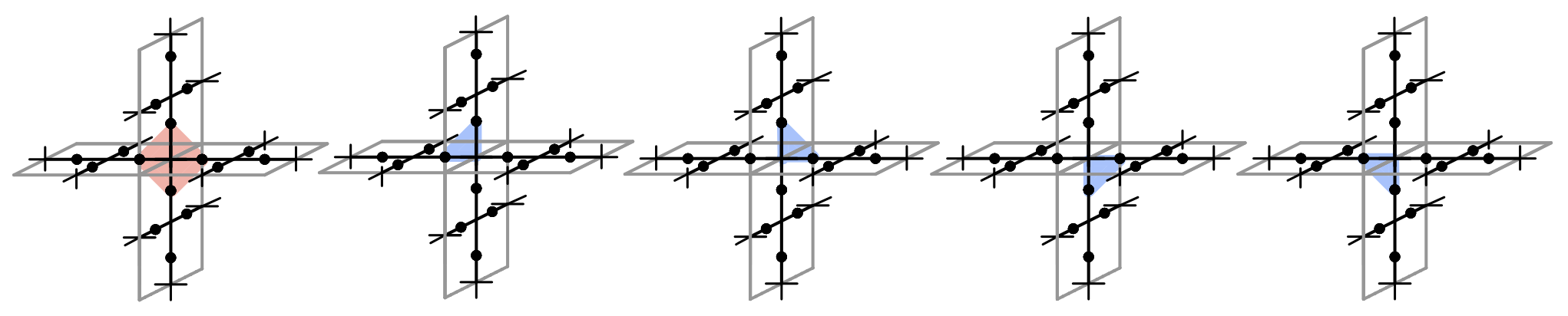}
    \includegraphics[width=0.95\columnwidth]{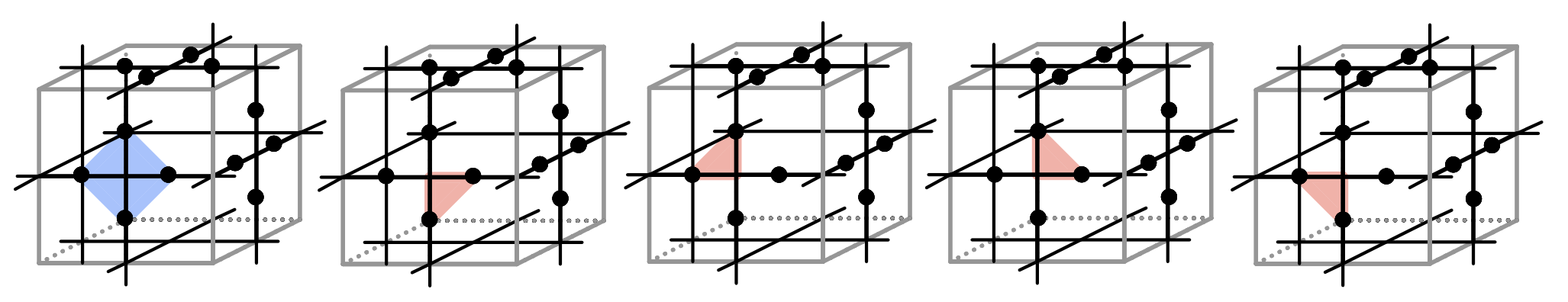}
    \caption{Subsystem code for the 4-star complex. Qubits live on bold edges of the black lines. To make illustration easier, we express the code in terms of a reference cubic lattice depicted in gray. (top) $Z$ and $X$ stabilizers correspond to vertices and cells of the cubic lattice. (middle) $X$ and $Z$ gauge generators on edges of the cubic lattice, in blue and red, respectively. (bottom) $X$ and $Z$ gauge generators on faces of the cubic lattice, in blue and red, respectively. Each stabilizer can be obtained by taking a product of either a set of face or edge gauge operators. }
    \label{fig:4star_subsystem_code}
\end{figure}

\textbf{Example: Bipartite $n$-star complex.} Take any 3d cell complex. Construct the fusion complex as per Appendix~\ref{app:higherdim}, consisting of GHZ and GHX resource states. This forms a valid subsystem toric code complex.

\textbf{Example: Enneahedral complex.}
Consider the complex in Fig,~\ref{fig:enneahedral_subsystem_code}, it has two types of cells, one with 9 square faces (enneahedra), and one with three square faces (trihedra).

\begin{figure}
    \centering
    \includegraphics[width=0.6\columnwidth]{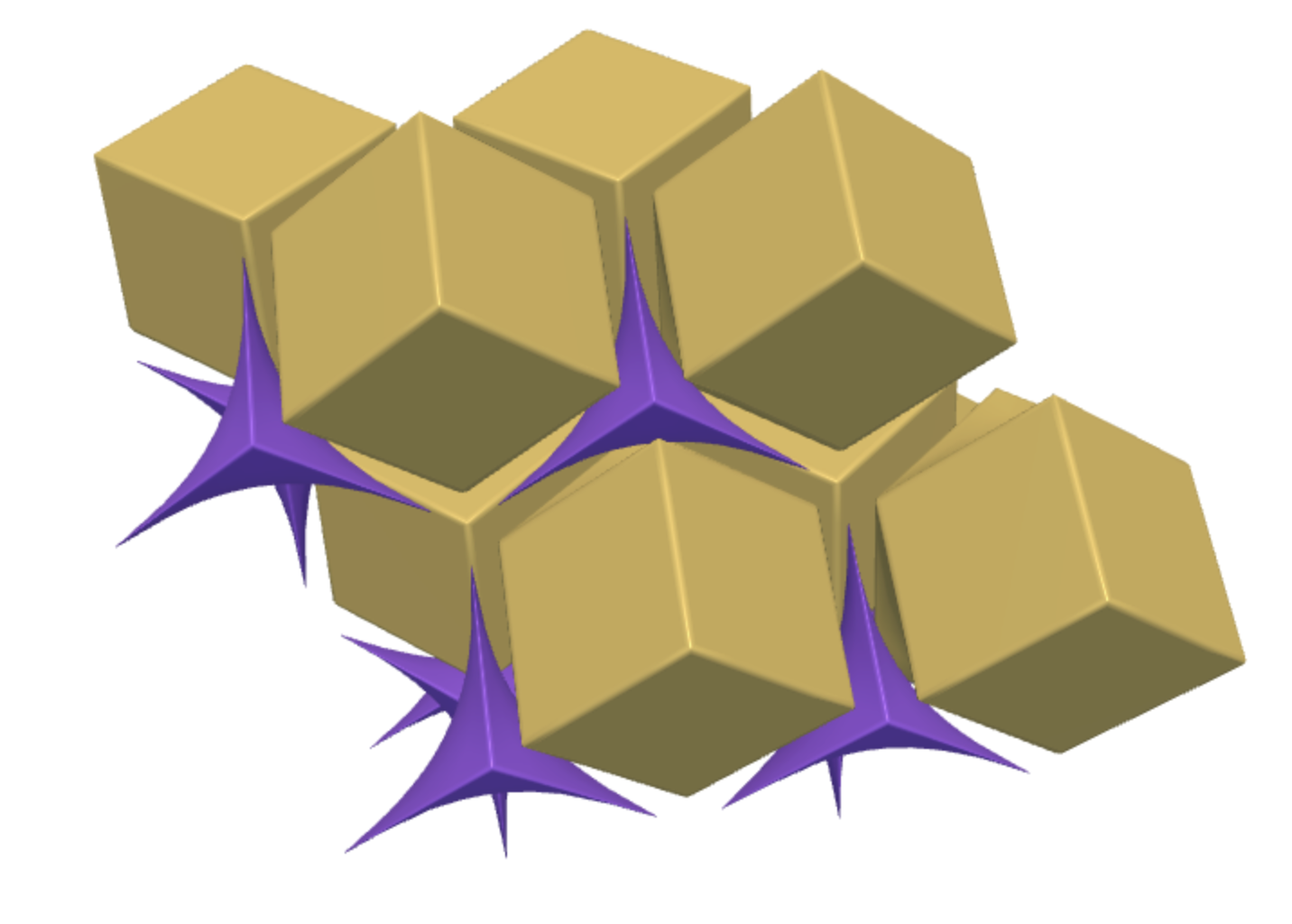}
    \includegraphics[width=0.7\columnwidth]{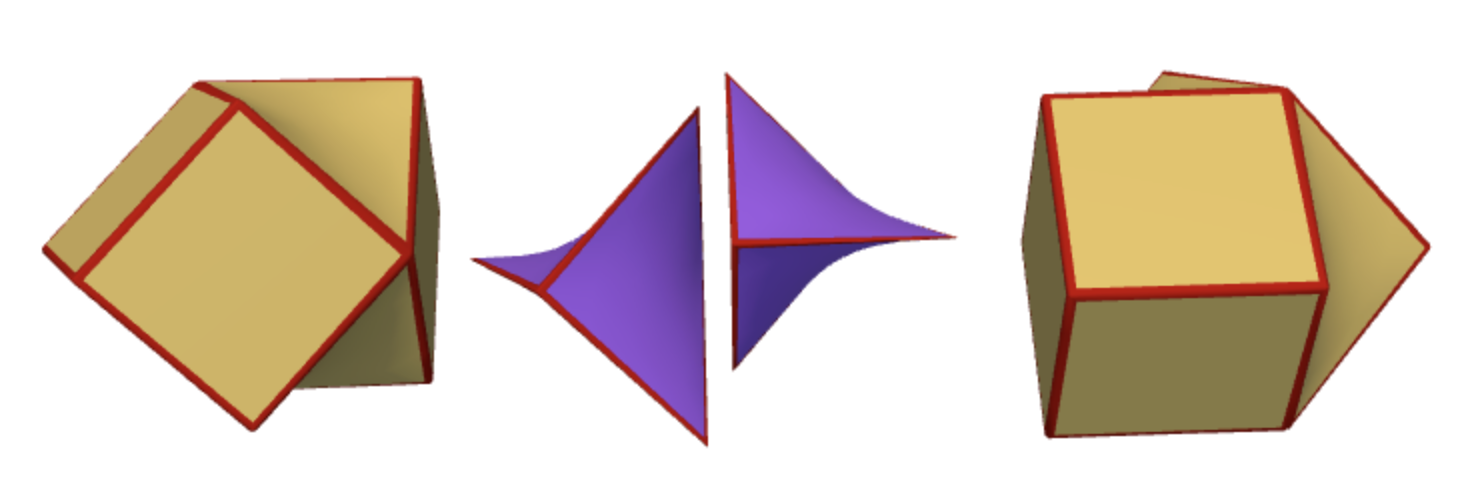}
    \caption{The enneahedral complex is a valid subsystem code complex. Gauge generators are weight 2, 3, and 4. Checks are weight 6 and 18. (top) The enneahedral complex. (bottom) Examples of the two types of cells. }
    \label{fig:enneahedral_subsystem_code}
\end{figure}

\section{Color-code fusion complexes}
\label{app:color_code}

The fusion-complex construction can be extended to include color-code based schemes~\cite{bombin2006topological,kubica2015unfolding,bombin20182d}, i.e., where the resource states are color codes on spheres, and the boundary of any fused region gives a color code state. 

Briefly, the construction is as follows. We work in the dual picture (the R-homology picture), where a color-code fusion complex can be defined on any 3-dimensional cell complex with some constraints, as follows.

\begin{definition}\label{def:color_code_complex} (Color-code fusion complex) A 3-dimensional cell complex, $\complexF = \{C, F, E, V\}$ is a color-code fusion complex, if \begin{itemize}
    \item The 0-cells are 3-colorable.
    \item The 2-skeleton is a triangulation (i.e., all faces are triangles). 
\end{itemize}
\end{definition}

In this picture, the color-code fusion complex corresponds to a fusion network with the following properties:
\begin{itemize}
    \item Resource states are associated to every 3-cell.
    \item Qubits and fusions are associated to every 2-cell. 
    \item An $X$-check and $Z$-check is associated to every 0-cell. 
\end{itemize}

In particular, the resource state on every 3-cell is a non-degenerate (meaning the stabilizers of the resource state specify a unique state) color code on a (topological) sphere. The faces host qubits of the resource state, and vertices of the 3-cell give a resource-state stabilizer of both $X$-type and $Z$-type (defined on all 2-cells incident to the vertex). Commutativity is guaranteed by the colorability of the vertices and the fact that every 2-cell is a triangle. Similarly to the surface code R-complex, where the resource state is a 2d dual of the plaquette-version of the surface code\footnote{The requirements in Def.~\ref{def:color_code_complex} can be compared to the conditions for surface-code fusion complexes in the R-homology complex (Def.~\ref{def:homologycomplex}), where every 2-cell is a square (and the vertices being locally 2-colorable follow from this).}; here the resource state is a 2d dual of the original topological color code of Ref.~\cite{bombin2006topological}, see Fig.~\ref{fig:colorcode}, for example.

Fusions ($\{ XX, ZZ \}$) occur between qubits on neighbouring resource states that share a face. Checks are formed as the product of all $X$-type fusions (or $Z$-type fusions) on the faces incident to a given vertex. This is clearly in the measurement group, and it can straightforwardly be verified that it is also in the resource state group.

Color-code fusion complexes define color-code networks, in that after performing fusions in any region of the network, the state on the boundary is a color code (up to signs of the stabilizers). One can also define 3d gauge color codes from the color-code fusion complex, where, much like the subsystem toric code, the gauge color code can also be understood as a collection of intersecting 2d color codes~\cite{roberts2020symmetry,bombin2015gauge}. Pursuing fusion complex constructions for other computational schemes~\cite{koenig2010quantum,ellison2022pauli,roberts20203,bauer2023topological,delfosse2013tradeoffs,higgott2023constructions} may be interesting. 

\begin{figure}
    \centering
    \includegraphics[width=0.4\columnwidth]{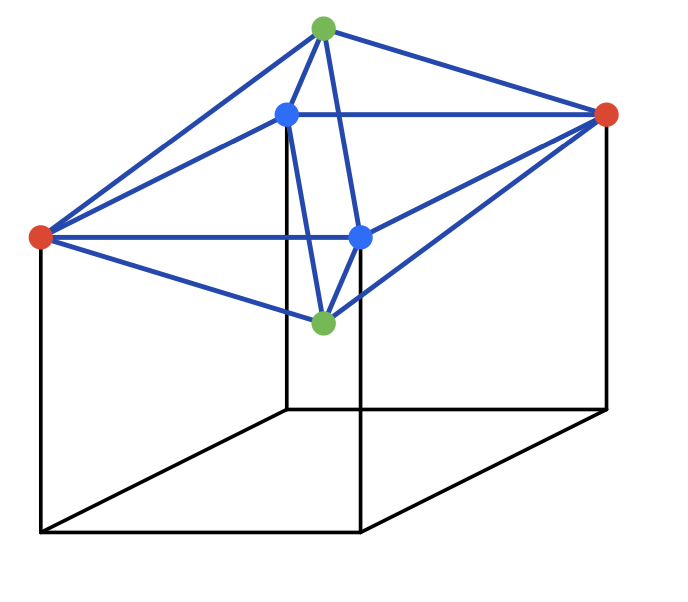} \qquad 
    \includegraphics[width=0.4\columnwidth]{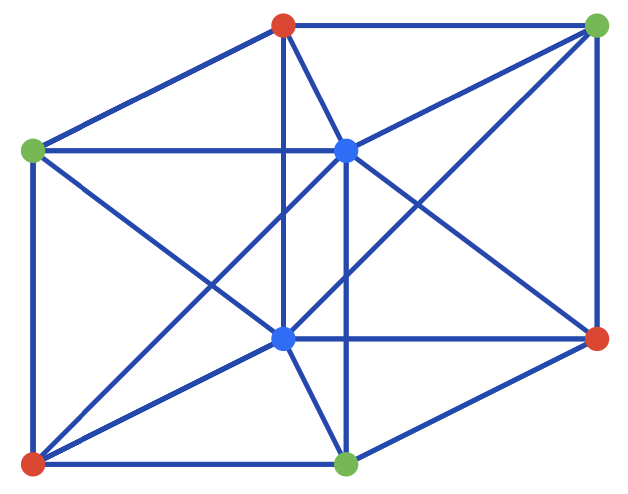}
    \caption{Two examples of color-code fusion networks in the dual picture. (left) Unit cell for complex with 8-qubit resource states, and checks of weight 12 and 24. The resource state forms a square bipyramid, tiling a bcc lattice (only one resource state cell is depicted).
    (right) Unit cell for complex with 12-qubit resource states, and uniform checks of weight 18. The resource state forms a hexagonal bipyramid.}
    \label{fig:colorcode}
\end{figure}

\section{Table of Examples}

Table~\ref{tab:examples_longlist} gives an extended list of examples with D-symbols up to size 8.

\newpage

\bibliographystyle{unsrt}
\bibliography{fbqc}

\newpage 

\begin{table}[]
\begin{tabular}{lll}
\hline
D-Symbol                                                                                                                                                                     & C                  & R            \\ \hline
\textless{}1 3:1,1,1,1:4,3,4\textgreater{}                                                                                                                                         & 12                 & 6            \\ \hline
\textless{}2 3:2,1 2,1 2,2:6,2 4,4\textgreater{}                                                                                                                                   & 24                 & 3x4          \\ \hline
\textless{}2 3:1 2,1 2,1 2,2:3 3,3 4,4\textgreater{}                                                                                                                               & 2x6+12             & 12           \\ \hline
\textless{}4 3:2 4,1 2 3 4,3 4,2 4:4 6,2 6,4\textgreater{}                                                                                                                         & 2x24               & 4x6          \\ \hline
\textless{}4 3:1 2 3 4,1 2 4,1 3 4,2 4:3 3 6,3 3,4\textgreater{}                                                                                                                   & 2x6+2x18           & 4x6          \\ \hline
\textless{}4 3:1 2 3 4,1 2 4,1 3 4,2 4:4 4 6,2 3,4\textgreater{}                                                                                                                   & 3x4+36             & 6x4          \\ \hline
\textless{}4 3:2 4,1 2 3 4,1 2 3 4,3 4:4 4,2 4 4 3,4 4\textgreater{}                                                                                                               & 3x8+24             & 3x4+12       \\ \hline
\textless{}4 3:1 2 3 4,1 2 3 4,1 3 4,2 4:3 3 4 4,4 4 2,4\textgreater{}                                                                                                             & 3x4+12+24          & 3x8          \\ \hline
\textless{}5 3:1 3 5,2 3 4 5,1 4 5,1 2 3 5:3 4,3 4,4 4\textgreater{}                                                                                                               & 2x10               & 10           \\ \hline
\textless{}6 3:1 2 3 5 6,1 2 4 6,1 3 4 5 6,2 6 5:4 4 4,2 3 4,4 4\textgreater{}                                                                                                     & 4+20               & 2x6          \\ \hline
\textless{}6 3:1 2 3 5 6,1 2 4 6,1 3 4 5 6,2 6 5:4 4 12,2 3 2,4 4\textgreater{}                                                                                                    & 3x4+60             & 12x3         \\ \hline
\textless{}6 3:1 2 3 4 5 6,1 3 5 6,2 3 6 5,1 4 5 6:3 4 4 6,3 3,4 4\textgreater{}                                                                                                   & 2x9+18             & 3x6          \\ \hline
\textless{}6 3:1 2 3 4 6,1 2 3 5 6,2 4 5 6,1 3 4 5 6:4 3 3 3,4 3 3,4 4 4\textgreater{}                                                                                             & 8x6+24             & 6+3x10       \\ \hline
\textless{}6 3:1 2 3 4 6,1 2 3 5 6,2 4 5 6,1 3 4 5 6:4 4 4 3,2 6 3,4 4 4\textgreater{}                                                                                             & 6x4+48             & 6+3x10       \\ \hline
\textless{}7 3:1 2 3 4 5 7,1 3 4 6 7,2 3 5 6 7,4 2 3 5 6 7:3 8 3 3,3 3 3,4 4 4 4\textgreater{}                                                                                     & 8x6+36             & 7x6          \\ \hline
\textless{}8 3:1 2 3 4 5 6 7 8,2 4 6 8,2 5 7 8,3 4 6 8:6 6 4 6,2 3,4\textgreater{}                                                                                                 & 4x6+2x36           & 12x4         \\ \hline
\textless{}8 3:2 4 6 8,2 3 4 5 6 7 8,3 4 7 8,2 5 6 8:3 4 4 4,4 2 4,4\textgreater{}                                                                                                 & 6x8+2x24           & 6x8          \\ \hline
\textless{}8 3:2 4 6 8,1 2 3 4 5 6 7 8,3 4 7 8,2 5 6 8:4 4 4 6,2 8 8 2,4\textgreater{}                                                                                             & 3x16+48            & 6x8          \\ \hline
\textless{}8 3:2 4 6 8,1 2 3 4 5 6 7 8,3 4 7 8,2 5 6 8:4 8 8 6,2 4 4 2,4\textgreater{}                                                                                             & 3x16+48            & 12x4         \\ \hline
\textless{}8 3:2 4 6 8,1 2 3 4 8 7,1 2 5 6 7 8,3 4 8 7:6 6 4,2 3 3 3,4 4\textgreater{}                                                                                             & 2x12+4x18          & 6x4+4x6      \\ \hline
\textless{}8 3:2 4 6 8,3 2 4 7 6 8,1 2 3 4 7 8,5 6 7 8:4 4,4 2 4 2 6,4 4\textgreater{}                                                                                             & 4x12+3x16          & 6x4+24       \\ \hline
\textless{}8 3:2 4 6 8,3 2 4 7 6 8,1 2 3 4 7 8,5 6 7 8:4 4,6 2 3 2 6,4 4\textgreater{}                                                                                             & 4x12+2x24          & 4x6+2x12     \\ \hline
\textless{}8 3:2 4 6 8,3 2 4 7 6 8,1 2 3 4 7 8,5 6 7 8:4 4,6 2 4 2 4,4 4\textgreater{}                                                                                             & 6x8+48             & 4x6+3x8      \\ \hline
\textless{}8 3:2 4 6 8,1 2 3 4 5 6 7 8,3 4 7 8,5 6 7 8:4 8 4 8,2 4 6 2,4 4\textgreater{}                                                                                           & 3x16+48            & 6x4+4x6      \\ \hline
\textless{}8 3:1 2 4 5 7 8,1 3 4 6 7 8,2 3 4 8 7,1 5 6 7 8:3 3 3 3,3 3 4,4 4\textgreater{}                                                                                         & 4x6+2x12           & 2x12         \\ \hline
\textless{}8 3:2 4 6 8,1 2 3 4 7 6 8,1 2 5 6 7 8,3 4 7 8:4 4 4,2 2 3 4 4,4 4\textgreater{}                                                                                         & 2x4+24             & 2x4+8        \\ \hline
\textless{}8 3:2 4 6 8,1 2 3 4 7 6 8,1 2 5 6 7 8,3 4 7 8:4 4 4,3 2 3 4 3,4 4\textgreater{}                                                                                         & 2x6+2x18           & 4x6          \\ \hline
\textless{}8 3:2 4 6 8,1 2 3 4 7 6 8,1 2 5 6 7 8,3 4 7 8:4 4 8,3 4 3 2 4,4 4\textgreater{}                                                                                         & 4x24               & 6x4+4x6      \\ \hline
\textless{}8 3:2 4 6 8,1 2 3 4 7 6 8,1 2 5 6 7 8,3 4 7 8:6 6 8,2 3 3 2 3,4 4\textgreater{}                                                                                         & 12+36              & 4x3+3x4      \\ \hline
\textless{}8 3:2 4 6 8,3 2 4 7 6 8,1 2 3 4 5 6 8,5 6 7 8:4 4,4 2 3 3 3,4 4 4\textgreater{}                                                                                         & 4x12               & 2x3+2x9      \\ \hline
\textless{}8 3:2 4 6 8,1 2 3 4 7 6 8,1 2 5 6 7 8,3 4 7 8:4 4 12,2 4 3 2 4,4 4\textgreater{}                                                                                        & 3x8+72             & 12x4         \\ \hline
\textless{}8 3:2 4 6 8,1 2 3 4 7 6 8,1 2 5 6 7 8,3 4 7 8:4 4 12,3 3 3 2 3,4 4\textgreater{}                                                                                        & 2x12+2x36          & 8x3+4x6      \\ \hline
\textless{}8 3:1 2 3 5 6 8,1 2 4 5 7 8,1 3 5 6 7 8,2 6 7 8:3 3 3 3,2 4 4 3,4 4\textgreater{}                                                                                       & 2x3+2x9+2x12       & 2x12         \\ \hline
\textless{}8 3:1 2 3 5 6 8,2 4 5 7 8,1 3 4 5 6 7 8,2 6 7 8:6 3 3,4 4 2 2,4 4 4\textgreater{}                                                                                       & 12x3+60            & 3x4+6x6      \\ \hline
\textless{}8 3:1 2 3 4 5 6 7 8,1 2 3 4 5 7 8,2 4 6 8,3 5 8 7:3 3 3 3 3 4 3,2 4 4,4\textgreater{}                                                                                   & 4x3+12+24          & 3x8          \\ \hline
\textless{}8 3:2 4 6 8,1 2 3 4 5 6 7 8,1 2 5 6 7 8,3 4 7 8:4 4 4 4,2 4 8 2 3 4,4 4\textgreater{}                                                                                   & 3x8+3x16+24        & 6x4+24       \\ \hline
\textless{}8 3:2 4 6 8,1 2 3 4 5 6 7 8,1 2 5 6 7 8,3 4 7 8:4 4 4 4,3 3 2 4 3 4,4 4\textgreater{}                                                                                   & 6x8+2x12+24        & 8x3+24       \\ \hline
\textless{}8 3:2 4 6 8,1 2 3 4 5 6 7 8,1 2 5 6 7 8,3 4 7 8:4 4 6 6,2 4 8 2 2 4,4 4\textgreater{}                                                                                   & 3x8+24+48          & 6x4+3x8      \\ \hline
\textless{}8 3:2 4 6 8,1 2 3 4 5 6 7 8,1 2 5 6 7 8,3 4 7 8:4 4 6 6,3 3 2 6 3 2,4 4\textgreater{}                                                                                   & 2x12+24            & 4x3+12       \\ \hline
\textless{}8 3:2 4 6 8,1 2 3 4 5 6 7 8,1 2 5 6 7 8,3 4 7 8:4 4 8 8,3 4 2 4 3 2,4 4\textgreater{}                                                                                   & 3x16+2x24          & 8x3+3x8      \\ \hline
\textless{}8 3:2 4 6 8,1 2 3 4 5 6 7 8,1 2 5 6 7 8,3 4 7 8:6 6 8 8,2 3 4 2 2 3,4 4\textgreater{}                                                                                   & 2x12+24+48         & 8x3+6x4      \\ \hline
\textless{}8 3:2 4 6 8,3 2 4 7 6 8,1 2 3 4 5 6 7 8,5 6 7 8:4 4,2 3 3 4 3 4,4 4 4 4\textgreater{}                                                                                   & 8x6+48             & 6x4+2x6+12   \\ \hline
\textless{}8 3:2 4 6 8,3 2 4 7 6 8,1 2 3 4 5 6 7 8,5 6 7 8:4 4,2 4 4 8 2 3,4 4 4 4\textgreater{}                                                                                   & 6x8+48             & 3x4+3x8+12   \\ \hline
\textless{}8 3:1 2 3 5 6 8,1 2 4 5 7 8,1 3 4 5 6 7 8,2 6 7 8:3 3 3 3,3 3 3 6 3,4 4 4\textgreater{}                                                                                 & 5x6+18             & 6+18         \\ \hline
\textless{}8 3:1 2 3 5 6 8,1 2 4 5 7 8,1 3 4 5 6 7 8,2 6 7 8:3 3 3 3,4 6 4 2 2,4 4 4\textgreater{}                                                                                 & 12x3+12+48         & 3x4+3x12     \\ \hline
\textless{}8 3:1 2 3 5 6 8,1 2 4 5 7 8,1 3 4 5 6 7 8,2 6 7 8:4 4 3 3,2 3 4 8 3,4 4 4\textgreater{}                                                                                 & 3x4+6x8+36         & 4x12         \\ \hline
\textless{}8 3:1 2 3 5 6 8,1 2 4 5 7 8,1 3 4 5 6 7 8,2 6 7 8:4 4 3 3,3 3 4 8 2,4 4 4\textgreater{}                                                                                 & 6x8+4x12           & 3x4+36       \\ \hline
\textless{}8 3:1 2 3 5 6 8,1 2 4 5 7 8,1 3 4 5 6 7 8,2 6 7 8:4 4 6 6,2 3 3 2 4,4 4 4\textgreater{}                                                                                 & 3x4+3x12+48        & 12x3+12      \\ \hline
\textless{}8 3:1 2 3 5 6 8,1 2 4 5 7 8,1 3 4 5 6 7 8,2 6 7 8:4 4 9 9,2 3 2 2 4,4 4 4\textgreater{}                                                                                 & 3x4+36+48          & 12x3+3x4     \\ \hline
\textless{}8 3:1 2 4 5 6 8,1 3 4 5 7 8,2 3 4 6 8,5 2 3 4 6 7 8:3 3 3 3,3 3 4,4 4 4 4\textgreater{}                                                                                 & 8x6+4x12           & 6+3x14       \\ \hline
\textless{}8 3:1 2 3 4 5 6 7 8,1 2 4 5 7 8,1 3 5 6 7 8,2 6 7 8:3 3 4 4 4 4,3 4 2 3,4 4\textgreater{}                                                                               & 6x4+2x6+12+48      & 8x6          \\ \hline
\textless{}8 3:1 2 3 4 5 6 7 8,1 2 3 4 5 6 7 8,2 4 6 8,3 5 7 8:3 3 3 3 3 4 3 4,2 4 4 4,4\textgreater{}                                                                             & 8x3+2x12+2x24      & 6x8          \\ \hline
\textless{}8 3:1 2 4 5 6 8,1 3 4 5 7 8,2 3 4 6 7 8,5 2 3 4 6 7 8:3 3 3 6,3 3 3 3,4 4 4 4 4\textgreater{}                                                                           & 8x6+2x24           & 8x6          \\ \hline
\end{tabular}
\caption{Table of fusion complexes.}
\label{tab:examples_longlist}
\end{table}

\end{document}